\documentclass{article}
\usepackage{ijcai23}

\usepackage{times}
\usepackage{soul}
\usepackage{url}
\usepackage[hidelinks]{hyperref}
\usepackage[utf8]{inputenc}
\usepackage[small]{caption}
\usepackage{graphicx}
\usepackage{amsmath}
\usepackage{sgame}
\usepackage{amsthm}
\usepackage{enumitem}
\usepackage{booktabs}
\usepackage{algorithm}
\usepackage{algorithmic}
\usepackage[switch]{lineno}

\usepackage{amsfonts}
\usepackage{mathrsfs}
\usepackage{amssymb}
\usepackage{tikz}
\usetikzlibrary{arrows}
\usetikzlibrary{arrows.meta}
\usepackage{multirow}
\usepackage{graphics}
\usepackage{xspace}
\usepackage{comment}

\newtheorem{theorem}{Theorem}
\newtheorem{corollary}[theorem]{Corollary}
\newtheorem{lemma}[theorem]{Lemma}
\newtheorem{definition}[theorem]{Definition}
\newtheorem{proposition}[theorem]{Proposition}

\newcommand{\NP}{\textsf{NP}}
\newcommand{\PLS}{\textsf{PLS}}

\newcommand{\maxweightpartition}{\textsc{MaximumWeightedDigraphPartition}\xspace}
\newcommand{\mwdp}{\textsc{MWDP}\xspace}

\newcommand{\one}{\textsf{one}\xspace}
\newcommand{\two}{\textsf{two}\xspace}

\newcommand{\sbf}{\ensuremath{\mathbf{s}}\xspace}

\newcommand{\pot}{\Phi}
\newcommand{\Real}{\mathbb{R}}

\newcommand{\wcal}{\ensuremath{\mathcal{W}}\xspace}
\newcommand{\fcal}{\ensuremath{\mathcal{F}}\xspace}

\newcommand{\acal}{\ensuremath{\mathcal{A}}\xspace}

\newcommand{\pf}{\textbf{Proof:} }
\newcommand{\2}{\vspace{0.2cm}}

\title{Complexity of Efficient Outcomes in Binary-Action Polymatrix Games with Implications for Coordination Problems}

\author{
Argyrios Deligkas$^1$
\and
Eduard Eiben$^1$\and
Gregory Gutin$^1$\and
Philip R.\ Neary$^1$ \And
Anders Yeo$^{2,3}$
\affiliations
$^1$Royal Holloway, University of London\\
$^2$University of Southern Denmark\\
$^3$University of Johannesburg
\emails
\{argyrios.deligkas, eduard.eiben, g.gutin, philip.neary\}@rhul.ac.uk,
andersyeo@gmail.com
}

\begin{document}

\maketitle

\begin{abstract}
We investigate the difficulty of finding economically efficient solutions to coordination problems on graphs.
Our work focuses on two forms of coordination problem:\ pure-coordination games and anti-coordination games.
We consider three objectives in the context of simple binary-action polymatrix games:\ (i) maximizing welfare, (ii) maximizing potential, and (iii) finding a welfare-maximizing Nash equilibrium.
We introduce an intermediate, new graph-partition problem, termed \maxweightpartition, which is of independent interest, and we provide a  complexity dichotomy for it.
This dichotomy, among other results, provides as a corollary a dichotomy for Objective (i) for {\em general} binary-action polymatrix games.
In addition, it reveals that the complexity of achieving these objectives varies depending on the form of the coordination problem.
Specifically, Objectives (i) and (ii) can be efficiently solved in pure-coordination games, but are \NP-hard in anti-coordination games.
Finally, we show that objective (iii) is \NP-hard even for simple non-trivial pure-coordination games.
\end{abstract}

\section{Introduction}
\label{sec:intro}

A coordination problem is one wherein all parties can realize mutual gains, but only by making mutually consistent decisions.
Such problems can range from small-scale issues, such as arranging where to meet with a friend \cite{Schelling:1978:}, to larger institutional-level issues, such as ensuring the efficient functioning of organizations \cite{Young:2001:}.

Broadly speaking there are two kinds of coordination problems:\ ``{\em pure-coordination games}'', in which it is beneficial to act the same as others, and ``{\em anti-coordination games}'', where it pays to differentiate your behavior from that of others.
Examples of the former include using euros as currency because others use euros and driving on the same side of the road as everyone else.
Examples of the latter include not supplying a public good if a neighbor is already doing so, and choosing a product that is different from the mainstream.

We zoom in on the payoff-structure of coordination problems and focus on the following question.
\begin{quote}
    {\em Are there characteristics inherent to certain coordination problems that render them easier to resolve?}
\end{quote}


In this paper, we adopt the most fundamental framework of multi-player coordination, and model the two above-mentioned environments as binary-action {\em polymatrix games}~\cite{Janovskaya:1968:}.
The expressive power of polymatrix games has made them the ``go-to'' method to model problems ranging from coordination games on graphs~\cite{apt2017coordination,apt2022coordination,rahn2015efficient} and additively separable hedonic games~\cite{bogomolnaia2002stability}, to building-blocks for hardness reductions~\cite{CDT,DGP,DFHM22-focs,Rubinstein18-polymatrix} and applications in protein-function prediction and semi-supervised learning~\cite{elezi2018transductive-gtg,vascon2020protein-gtg}.

Formally, a polymatrix game is represented by a graph, where every vertex corresponds to a player, and every edge corresponds to a two-player game that is played between the adjacent vertices.
A player's payoff is the sum of payoffs earned from interacting with every player in their neighborhood, where the same action must be used with each. 
The graph structure captures the dependencies of the players, while the payoff structure of the two-player games models the {\em nature} of each pairwise interaction that can, in theory, vary arbitrarily:\ from settings of pure competition to those with perfectly aligned interests.
In pure-coordination polymatrix games, every two-player game has as pure Nash equilibria the two strategy profiles where the players choose the {\em same} action.
On the other hand, in anti-coordination polymatrix games, every two-player game has as pure Nash equilibria the two strategy profiles where the players choose {\em different} actions. See Fig.~\ref{fig:example} for $2\times2$ example of each.


\begin{figure}
\hspace*{\fill}\begin{game}{2}{2}[${i}$][${j}$]
& $\one$ & $\two$\\
$\one$ & $2,1$ & $0,0$\\
$\two$ &$0,0$ & $1, 2$
\end{game}\hspace*{\fill}
\begin{game}{2}{2}[${i}$][${j}$]
& $\one$ & $\two$\\
$\one$ & $9,9$ & $2,10$\\
$\two$ & $10,2$ & $0,0$
\end{game}\hspace*{\fill}
\caption[]{Examples of coordination games. Left: Pure-coordination game; Right: anti-coordination game.}
\label{fig:example}
\end{figure}

\subsection{Our Contribution}
We provide a comprehensive study for the complexity-landscape of economically-efficient outcomes for binary-action pure-coordination and anti-coordination polymatrix games.
We focus on the computational complexity of the following objectives for these classes of games.
\begin{description}
    \item[Objective (i):] maximize social welfare.
    \item[Objective (ii):] maximize a potential function; when the game is pairwise-potential.
    \item[Objective (iii):] find a welfare-optimal Nash equilibrium. 
\end{description}


We will show that Objectives (i) and (ii) are special cases of a novel graph-partition problem, that we call \maxweightpartition (\mwdp for short).
This problem is of independent interest, since it includes many well-known graph-theoretic problems like maximum and minimum cut~\cite{korte2011combinatorial} as special cases.
Our first technical result is a complexity dichotomy for \mwdp (Theorem~\ref{thm:dichotomy}).
This result yields as an immediate corollary a dichotomy for social-welfare outcomes for {\em general} binary-action polymatrix games, a result which, to the best of our knowledge, was not known before.

The dichotomy for \mwdp reveals, again as immediate corollary, that anti-coordination games are significantly harder than pure-coordination games. 
More specifically, it shows that Objective (i) can be solved in polynomial time for pure-coordination games, while it is \NP-hard in the worst case for anti-coordination games. 
Although the \NP-hardness for anti-coordination games was implied by \citeauthor{CD11-Cai-Daskalakis}~\shortcite{CD11-Cai-Daskalakis}, our result provides a more fine grained resolution for this problem, since it identifies a much larger class of games for which the problem is intractable.

Another corollary of our dichotomy is the tractability of Objective (ii) for pairwise-potential pure-coordination games, i.e., pure-coordination polymatrix games where every two-player game admits a potential. 
This result is in stark contrast to pairwise-potential anti-coordination games where it is known that not only is it \NP-hard to find a potential-maximizing outcome, but it is \PLS-complete, i.e., intractable, to find {\em any} local maximum of the potential function~\cite{CD11-Cai-Daskalakis}.

Given the positive results for Objectives (i) and (ii) for pure-coordination games, one might wonder if Objective (iii) is tractable as well for this type of game.
Observe that for these games it is easy to find an arbitrary Nash equilibrium: every such game possesses at least two ``trivial'' Nash equilibria wherein all players choose the same action.
In addition, every potential-maximizing outcome corresponds to a Nash equilibrium too.
Unfortunately, as our second technical result shows, Objective (iii) becomes immediately intractable for almost every potential pure-coordination game where one of the two trivial Nash equilibria is not an obviously-optimal solution.
In fact, we provide a dichotomy for Objective (iii) for the arguably most fundamental subclass of pure-coordination games, known as {\em threshold games}~\cite{NearyNewton:2017:JMID}.

\subsection{Further Related Work}

Our work relates to several areas of study, including the significance of coordination problems and the role of maximizing potential in economic contexts, and issues of computational complexity in algorithmic game theory.

\paragraph{Coordination problems in economics.}
Since they were introduced in philosopher David Lewis' study on convention and language \cite{Lewis:1969:}, coordination games have been one of the modeling tool of choice for economists, being applied to the adoption of technological standards \cite{KatzShapiro:1985:AER,FarrellSaloner:1985:RJE,Arthur:1989:EJ}, the setting of macroeconomic policy \cite{CooperJohn:1988:QJE}, the study of bank runs \cite{DiamondDybvig:1983:JPE}, etc.


Pure coordination polymatrix games in particular have received lots of attention from game theorists.
A vast literature considers equilibrium selection in the case where every two-player game is a ``stag-hunt'' and finds that, for almost all network structures, uniform adoption of the ``safe'' risk-dominant action is the long run prediction \cite{KandoriMailath:1993:E,FosterYoung:1990:TPB,Young:1993:E,Ellison:1993:E,Morris:2000:RES,Peski:2010:JET}.
The conclusion of this literature is that even if there is a universally agreed upon optimal equilibrium, successfully achieving that outcome is far from assured.


In the above papers, all players are in a sense ``the same'', since a common $2\times2$ game occurs along every edge of the graph.
Pure coordination polymatrix games with heterogeneous preferences, in particular the {\em language game}~\cite{Neary:2012:GEB} and the {\em threshold model}~\cite{NearyNewton:2017:JMID}, will play an important role in our analysis.
We defer a detailed description of these games to Section~\ref{sec:best-NE}.

Though of similar importance, anti-coordination polymatrix games, see \citeauthor{Bramoulle:2007:GEB}~\shortcite{Bramoulle:2007:GEB}, have certainly received less attention.
Other network games wherein action choice is a strategic substitute include those of public good provision \cite{BramoulleKranton:2007:JET,GaleottiGoyal:2010:RES}.

Coordination problems are of such economic importance that countless experiments have been performed to try and ascertain how people attempt to coordinate on optimal outcomes and when they will be successful.
\citeauthor{Van-HuyckBattalio:1990:AER}~\shortcite{Van-HuyckBattalio:1990:AER} find that smaller groups successfully coordinate far more frequently than larger groups.
\citeauthor{KearnsSuri:2006:S}~\shortcite{KearnsSuri:2006:S} and \citeauthor{McCubbinsPaturi:2009:}~\shortcite{McCubbinsPaturi:2009:} consider coloring problems (i.e., anti-coordination games) on a variety of different network structures.
Both conclude that certain network structures, in particular ``small worlds'' networks, are easier for subjects to color successfully.

\paragraph{Potential-maximizing equilibria.}
Potential games were first introduced by \citeauthor{ShapleyMonderer:1996:GEB}~\shortcite{ShapleyMonderer:1996:GEB} and have received significant attention.
Potential-maximizing Nash equilibria are desirable since they are stochastically-stable \cite{Blume:1993:GEB}, can be uniquely absorbing \cite{HofbauerSorger:1999:JET}, select risk-dominant outcomes in games played on random networks \cite{Peski:arxiv}, and are robust in games with incomplete information \cite{Ui:2001:E}.
Potential games appear frequently in applied work, for example to study the effects of price discrimination policies in oligopolies \cite{ArmstrongVickers:2001:RAND}, the impact of uncertainty on technology adoption \cite{OstrovskySchwarz:2005:RAND}, and issues of collective action \cite{MyattWallace:2009:EJ}.
Many of the classic models in applied game theory are potential games including the Cournot model and congestion games~\cite{Rosenthal:1973:IJGT}.

%


\paragraph{Nash equilibria in polymatrix games.}
Computational aspects of (approximate) Nash equilibria in polymatrix games received a lot of attention over the years, starting from classical results from fifty years ago~\cite{eaves1973polymatrix,howson1972equilibria,howson1974bayesian,miller1991copositive}, to more recent results~\cite{CDT,DGP,Rubinstein18-polymatrix}, and the very recent dichotomy for the complexity of finding {\em any} approximate well-supported Nash equilibrium in general binary-action polymatrix games~\cite{DFHM22-focs}. 
\citeauthor{DBLP:conf/innovations/BoodaghiansKM20}~\shortcite{DBLP:conf/innovations/BoodaghiansKM20} studied the smoothed complexity on coordination polymatrix games, while~\citeauthor{aloisio2021distance}~\shortcite{aloisio2021distance} studied an extension of polymatrix games. The complexity of constrained Nash equilibria for general polymatrix games were studied in~\cite{DeligkasFS17-constrained}, while~\cite{barman2015approximating,DeligkasFS20-tree-polymatrix,EGG-elkind2006nash,ortiz2017tractable} study games with tree-underlying structure. Finally,~\cite{DeligkasFIS16-experimental} provides an experimental comparison of various algorithms for polymatrix games.

\section{Preliminaries}
\label{sec:prelims}

Given a $2\times 2$ matrix $M$, we denote by $m_{11}$, $m_{12}$, $m_{21}$ and 
$m_{22}$ the entries of $M$, such that $M = \left[\begin{array}{cc}
m_{11} & m_{12} \\
m_{21} & m_{22} \\
\end{array}\right]$.
For any graph $G$, we denote $V(G)$ and $E(G)$ the sets of vertices and edges respectively; if the graph is directed, we use $A(G)$ instead to denote the set of arcs of the graph. 
We assume knowledge of standard notions of directed graphs~\cite{BJG}. 

An $n$-player binary-action \emph{polymatrix game} is defined by a graph $G$, where each vertex represents a player.
Each player $i \in V(G)$ has two actions called \one and \two. For each edge $ij \in E(G)$, there is a $2 \times 2$ two-player game $(\Pi^{ij}, \Pi^{ji})$, where matrix $\Pi^{ij}$ gives the {\em payoffs} that player $i$ obtains from their interaction with player $j$, and likewise matrix $\Pi^{ji}$ gives the payoffs  player $j$ gets with the interaction with player $i$. 

A pure \emph{strategy profile} $\sbf = (s_1, s_2, \dots, s_n)$ specifies an action for each of the players; we will use $S$ to denote the set of all strategy profiles.
It is convenient to think of $s_i=(1,0)^{T}$ when player $i$ chooses action \one and $s_i=(0,1)^{T}$ when they choose action \two.
For each strategy profile $\sbf \in S$, the {\em payoff} of player $i$ is 
$p_i(\sbf) := s_i^{T} \cdot \sum_{j \; :\; ij \in E(G)} \Pi^{ij} \cdot s_j$.
In other words, the payoff obtained by a player is the sum of the payoffs obtained from the interaction with every neighboring player, where the same action must be used with each.

We are interested in computing welfare-optimal strategy profiles and in finding pure strategy Nash equilibrium profiles (and in the combination of both concepts).

\begin{definition}
The {\em social welfare} of strategy profile $\sbf$ is $\wcal{(\sbf{})} := \sum_{i \in V} p_i(\sbf{})$.
\end{definition}

We denote by $\sbf_{-i}$ the \emph{partial strategy profile} consisting of the strategies of the $n-1$ players other than $i$.
That is, from the perspective of player $i$ using strategy $s_i$, the strategy profile $\sbf{}$ can be viewed as $(s_i, \sbf_{-i})$.

A strategy profile $\sbf{}^*$ is a pure {\em Nash equilibrium} if no player can strictly increase their payoff by unilaterally changing their strategy choice. 
Formally, $\sbf^*$ is a Nash equilibrium, if for every player $i$ and every $s_i \neq s^*_i$ it holds that $p_i(s^*_i,\sbf_{-i}^*) \geq p_i(s_i,\sbf_{-i}^*)$. 




\subsection{Classes of polymatrix games}

We are focused on two classes of binary-action polymatrix games that capture coordination problems.

\paragraph{\bf Pure coordination games.}
In a {\em pure-coordination polymatrix game}, the payoff of a player increases with the number of neighbors who choose the {\em same} action as them. 
Formally, a binary-action polymatrix game is pure-coordination if for every edge $ij \in E(G)$ the strategy profiles $(\one, \one)$ and $(\two, \two)$ are Nash equilibria for the two-player game $(\Pi^{ij}, \Pi^{ji})$.
Observe that every game in this class possesses at least two ``trivial'' Nash equilibria wherein all players choose the same action.

\paragraph{\bf Anti-coordination games.} 
In an {\em anti-coordination polymatrix game}, each player's payoff increases with the number of neighbors who choose a {\em different} action.
Formally, a binary-action polymatrix game is anti-coordination if for every edge $ij \in E$ the strategy profiles $(\one, \two)$ and $(\two, \one)$ are Nash equilibria for the two-player game $(\Pi^{ij}, \Pi^{ji})$.


\paragraph{\bf Potential games.}
A strategic game is a {\em potential} game \cite{ShapleyMonderer:1996:GEB} if the incentive of all players to change their strategy can be expressed using a single function called the potential function.
Potential games possess many desirable properties:\ pure strategy equilibria correspond to local optima of the potential function so the existence of a pure strategy equilibrium is assured.
Formally, a game is a potential game if there exists a function $\pot: S \to \Real$ such that for every player $i$, for all $\sbf_{-i}$, and all pairs of actions $s_i', s_i'' \in S$
\[
\pot(s_i', \sbf_{-i}) - \pot(s_i'', \sbf_{-i}) = p_i(s_i', \sbf_{-i}) - p_i(s_i'', \sbf_{-i}).
\]
We emphasize that the same function $\pot$ captures the change in payoff associated with a deviation for {\em every} player.

\paragraph{\bf Pairwise-potential polymatrix games.}
When every pairwise interaction, i.e., two-player game, of a polymatrix game is a potential game, then the polymatrix game inherits this property. In fact, the potential at any strategy profile is equal to the sum of the potentials along every edge of the graph $G$. In other words, if $\Phi^{uv}$ denotes the pairwise-potential function for the two-player game played between $u$ and $v$, then the function $\Phi(\sbf) := \sum_{uv}\Phi^{uv}(\sbf)$ is a potential function for the polymatrix game.



\section{\maxweightpartition}
\label{sec:maxweightpartition-prob}
 In this section we introduce the \maxweightpartition problem (\mwdp), and we provide a dichotomy for its complexity.
 We then show how Objectives (i) and (ii) in binary-action polymatrix games are special cases of \mwdp, and as such complexity dichotomies can be given for each of these issues in relation to polymatrix games.

\paragraph{\bf $\maxweightpartition(\fcal)$.}
Given a family \fcal of $2 \times 2$ matrices, an instance of $\mwdp(\fcal)$ is given by a tuple $(D, c, f)$, where:
\begin{itemize}
    \item $D$ is an oriented graph on $n$ vertices (that is, a directed graph without any $2$-cycle);
    \item $c: A(D) \rightarrow \mathbb{R}^+$ assigns positive weights to arcs; 
    \item $f: A(D) \rightarrow \fcal$, is an assignment of a matrix from the family of matrices $\fcal$ to each arc in $D$.
\end{itemize}
Given a partition, $P=(X_1,X_2)$ of $V(D)$, the weight of an arc $uv \in A(D)$ is defined as follows, where $M=f(uv)$ is the matrix in $\fcal$
assigned to the arc $uv$.
\begin{equation*}
w^P(uv) = \left\{
\begin{array}{rcl}
c(uv) \cdot m_{11} & \mbox{if} & u,v \in X_1 \\
c(uv) \cdot m_{22} & \mbox{if} & u,v \in X_2 \\
c(uv) \cdot m_{12} & \mbox{if} & u \in X_1 \mbox { and } v \in X_2 \\
c(uv) \cdot m_{21} & \mbox{if} & u \in X_2 \mbox { and } v \in X_1 \\
\end{array}
\right.
\end{equation*}
Given a partition $P$, the weight of $D$, denoted by $w^P(D)$, is defined as the sum of the weights on every arc $uv$.
That is, $w^P(D) = \sum_{a \in A(D)} w^P(a)$.  The goal is to find a partition $P$ that maximizes $w^P(D)$.


We would like to highlight that the orientation of an arc in an instance of $\mwdp(\fcal)$ is required to determine which vertex defines the row and which vertex defines the column of the assigned matrix.

\subsection{A dichotomy for the complexity of $\mwdp(\fcal)$}
Our main result shows that the tractability of solving \maxweightpartition depends on properties of the matrices in the family of matrices $\fcal$.
We now introduce three properties that a matrix $M \in \fcal$ may satisfy. 
\begin{itemize}
    \item {\bf Property (a):} $m_{11} + m_{22} \geq m_{12} + m_{21}$.
    \item {\bf Property (b):} $m_{11} = \max\{ m_{11}, m_{22}, m_{12}, m_{21} \}$.
    \item {\bf Property (c):} $m_{22} = \max\{ m_{11}, m_{22}, m_{12}, m_{21} \}$.
\end{itemize}

Using the three properties above, we present a dichotomy for the complexity of $\mwdp(\fcal)$ with respect to \fcal.

\begin{theorem}
\label{thm:dichotomy}
An instance $(D,c, f)$ of $\mwdp(\fcal)$ can be solved in polynomial time if one of the following holds.
\begin{enumerate}
    \item \label{main-case1} All matrices in \fcal satisfy Property (a).
    \item \label{main-case2} All matrices in \fcal satisfy Property (b).
    \item \label{main-case3} All matrices in \fcal satisfy Property (c).
\end{enumerate}
In every other case, $\mwdp(\fcal)$ is \NP-hard.
\end{theorem}

Theorem~\ref{thm:dichotomy} decomposes the space into four cases, each of varying difficulty.
Cases~\ref{main-case2} and~\ref{main-case3} are immediate, since they admit trivially-optimal solutions: in Case~\ref{main-case2} all vertices belong to $X_1$ and in Case~\ref{main-case3} all vertices belong to $X_2$. 
On the other hand, Case~\ref{main-case1} is far from trivial and it requires a more sophisticated argument that creates an equivalent min-cut instance on an undirected graph (see Lemma~\ref{lem:case-1}), which is known to be solvable polynomial-time~\cite{korte2011combinatorial}. 
Finally, the last case deals with ``every other case'' and shows that the problem becomes intractable.
The proof involves a series of intricate subcases and constructions.\footnote{Due to space constraints, the formal proof appears in the Supplementary material, alongside other missing proofs.}

\begin{lemma}
\label{lem:case-1}
If all matrices in \fcal satisfy Property (a), then $\mwdp(\fcal)$ can be solved in polynomial time.
\end{lemma}

\begin{proof}
Let $(D,c,f)$ be an instance of $\mwdp(\fcal)$ that satisfies the constraints above.
We will construct a new edge-weighted, undirected, graph $H$ with vertex set $V(D) \cup \{s,t\}$ as follows. 
Let $UG(D)$ denote the undirected graph obtained from $D$ by removing orientations of all arcs. 
Initially, let $E(H) = E(UG(D)) \cup \{su,tu \; | \; u \in V(D) \}$ 
and let all edges in $H$ have weight zero. 
For each arc $uv \in A(D)$ we modify the edge-weight function $w$ of $H$ as follows, where $M=f(uv)$ is the matrix associated with the arc $uv$.

\begin{itemize}
\item Let $w(uv)=c(uv) \cdot (m_{11} + m_{22} - m_{12} - m_{21})/2$. 
Note $w(uv)$ is the weight of the undirected edge $uv$ in $H$ associated with the arc $uv \in A(D)$ and that this weight is non-negative; this is guaranteed by Property (a) and it is crucial for the correctness of the lemma.
\item Add $c(uv) \cdot (-m_{22})/2$ to $w(su)$.
\item Add $c(uv) \cdot (-m_{22})/2$ to $w(sv)$.
\item Add $c(uv) \cdot (m_{21}-m_{11}-m_{12})/2$ to $w(tu)$.
\item Add $c(uv) \cdot (m_{12}-m_{11}-m_{21})/2$ to $w(tv)$.
\end{itemize}

Let $\theta$ be the smallest possible weight of all edges in $H$ after we completed the above process ($\theta$ may be negative).
Now consider the weight function $w^*$ obtained from $w$ by 
subtracting $\theta$ from all edges incident with $s$ or $t$. That is, $w^*(uv)=w(uv)$ if $\{u,v\} \cap \{s,t\}=\emptyset$ and 
$w^*(uv)=w(uv) - \theta$ otherwise.  
Note that all $w^*$-weights in $H$ are non-negative.  
We will show that for any $(s,t)$-cut, $(X_1,X_2)$ in $H$, i.e., $(X_1,X_2)$ partitions $V(H)$ and $s \in X_1$ and $t \in X_2$, the $w^*$-weight of the cut is equal to $-w^P(D) - |V(D)|\cdot \theta$, where $P$ is the partition 
$(X_1 \setminus \{s\},X_2 \setminus \{t\})$ in $D$.  Therefore, a minimum-weight cut in $H$ maximizes $w^P(D)$ in $D$.

Let $(X_1,X_2)$ be any $(s,t)$-cut in $H$.  For every $u \in V(D)$ we note that exactly one of the edges $su$ and $ut$ will belong to the cut.
Therefore, we note that the $w^*$-weight of the cut is $|V(D)| \cdot \theta$ less than the $w$-weight of the cut.
 It therefore suffices to show that the $w$-weight of the cut is $-w^P(D)$ (where  $P$ is the partition
$(X_1 \setminus \{s\},X_2 \setminus \{t\})$ of $V(D)$).

Let $(X_1,X_2)$ be some $(s,t)$-cut. There are four possibilities for any $uv \in A(D)$, which follow from the definition of the $w$-weights.
\begin{itemize}
\item $u,v \in X_1$. In this case, we have added $-c(uv) \cdot m_{11}$ to the $w$-weight of the  $(s,t)$-cut.
\item $u,v \in X_2$. In this case, we have added $-c(uv) \cdot m_{22}$ to the $w$-weight of the  $(s,t)$-cut.
\item $u \in X_1$ and $v \in X_2$. In this case, we have added $-c(uv) \cdot m_{12}$ to the $w$-weight of the  $(s,t)$-cut.
\item $v \in X_1$ and $u \in X_2$. In this case, we have added $-c(uv) \cdot m_{21}$ to the $w$-weight of the  $(s,t)$-cut.
\end{itemize}

Therefore we note that in all cases we have added $-w^P(uv)$ to the $w$-weight of the  $(s,t)$-cut, $(X_1,X_2)$.
So the total $w$-weight of the  $(s,t)$-cut is $-w^P(D)$ as desired.

Analogously, if we have a partition $P=(X_1,X_2)$ of $V(D)$, then adding $s$ to $X_1$ and $t$ to $X_2$ we obtain a $(s,t)$-cut
with $w$-weight $-w^P(D)$ of $H$. 
As we can find a minimum $w^*$-weight cut in $H$ in polynomial time~\cite{korte2011combinatorial} we can find a  partition $P=(X_1,X_2)$ of $V(D)$ with 
minimum value of $-w^P(D)$, which corresponds to the maximum value of $w^P(D)$. Therefore, $\mwdp(\fcal)$ can be solved in polynomial time in this case.
\end{proof}


\section{Social Welfare and Potential Maximization}
\label{sec:sw-potent}
In this section we show how we can utilize Theorem~\ref{thm:dichotomy} and derive as corollaries several results, both positive and negative, for various classes of binary-action polymatrix games. First, we present a simple reduction from general binary-action polymatrix games to \mwdp that is used to solve Objective (i), i.e., maximize the social welfare.

\paragraph{\bf Social Welfare via \mwdp.}
Given a polymatrix game with underlying graph $G$, we create an instance of \mwdp as follows.
\begin{itemize}
    \item $D$ has the same vertex-set as $G$. In addition, for every edge of $G$ we add a directed edge in $D$ with arbitrary orientation. An edge from $u$ to $v$ specifies that player $u$ chooses a row and player $v$ chooses a column in the two-player game played between the two corresponding players.
    \item $c(uv)=1$ for every oriented edge $uv \in A(D)$.
    \item For each two-player game $(\Pi^{uv}, \Pi^{vu})$, with $uv \in E(G)$, we create the welfare-matrix $W^{uv}:=\Pi^{uv}+\Pi^{vu}$, and we associate it with the oriented edge between $u$ and $v$ as follows:
    \begin{itemize}
        \item $f(uv)=W^{uv}$, if we have added the arc $uv$;
        \item $f(uv)=(W^{uv})^T$, if we have added the arc $vu$.
    \end{itemize}
\end{itemize}
The reduction above induces an immediate translation between strategy profiles and partitions. Player $v$ chooses action \one if and only if the corresponding vertex $v$ in $D$ belongs to $X_1$. Conversely, player $v$ chooses action \two if and only if the corresponding vertex $v$ in $D$ belongs to $X_2$. 

For any strategy profile \sbf, we use $P(\sbf)$ to denote the corresponding partition. The following lemma trivially follows from the reduction above. 

\begin{lemma}
\label{lem:sw-to-mwdp}
For every possible strategy profile $\sbf$, it holds that $\wcal(\sbf) = w^{P(\sbf)}(D)$.
\end{lemma}

So, the combination of Lemma~\ref{lem:sw-to-mwdp} and Theorem~\ref{thm:dichotomy}, yields a series of results. The first one, is a clean complexity dichotomy for maximizing social welfare in general binary-action polymatrix games. To the best of our knowledge, this is the first dichotomy of this kind.

\begin{theorem}
\label{thm:sw-general}
Consider a binary-action polymatrix game on input graph $G$. Let $W^{uv} = \Pi^{uv}+\Pi^{vu}$ for every $uv \in E(G)$.
Finding a strategy profile that maximizes the social welfare can be solved in polynomial time if one of the following holds.
\begin{itemize}
   \item $w^{uv}_{11} + w^{uv}_{22} \geq w^{uv}_{12} + w^{uv}_{21}$ for every $uv \in E(G)$.
   \item $w^{uv}_{11} \geq \max \{w^{uv}_{12}, w^{uv}_{21}, w^{uv}_{22} \}$ for every $uv \in E(G)$.
   \item $w^{uv}_{22} \geq \max \{w^{uv}_{11}, w^{uv}_{12}, w^{uv}_{21} \}$ for every $uv \in E(G)$.
\end{itemize}
In every other case, the problem is \NP-hard.
\end{theorem}


The reduction from welfare maximization to \mwdp is immediate. The reduction in the opposite direction can be seen as follows. Consider an arc $uv$ of weight $c$ with matrix $M$. For entry $m_{ij}$ in $M$, let $\frac{c\cdot m_{ij}}{2}$ be the payoff to each player when player $u$ chooses strategy $i$ and player $v$ chooses $j$. Clearly the total welfare is given by $c\cdot m_{ij}$. The proof now follows by appealing to Theorem~\ref{thm:dichotomy}.

Theorem~\ref{thm:sw-general} implies two immediate corollaries for coordination polymatrix games. Firstly, for every welfare-matrix $W^{uv}$ of a pure-coordination game it holds that $w^{uv}_{11} + w^{uv}_{22} \geq w^{uv}_{12} + w^{uv}_{21}$. Second, it is not difficult to construct anti-coordination games where $w^{uv}_{11} + w^{uv}_{22} < w^{uv}_{12} + w^{uv}_{21}$.\footnote{The construction of Appendix C.2 from \citeauthor{CD11-Cai-Daskalakis}~\shortcite{CD11-Cai-Daskalakis} creates these types of welfare matrices.}

\begin{corollary}
\label{cor:pure-coord}
Maximizing social welfare in binary-action pure-coordination polymatrix games can be done in polynomial time.
\end{corollary}

\begin{corollary}
\label{cor:anti-coord}
It is \NP-hard to maximize social welfare in binary-action anti-coordination polymatrix games.
\end{corollary}


\paragraph{Pairwise-Potential Games via \mwdp.} Next we show how to reduce the problem of maximizing the potential in pairwise-potential binary-action polymatrix games to the \mwdp problem. 
Observe that for this class of games, every pairwise-potential $\Phi^{uv}$ can be written as a $2 \times 2$ potential-matrix, so 
$\Phi^{uv} := \left[\begin{array}{cc}
\phi^{uv}_{11} & \phi^{uv}_{12} \\
\phi^{uv}_{21} & \phi^{uv}_{22} \\
\end{array}\right]$. 
To create an instance of \mwdp we will follow a similar approach as before.
This time though, we use potential-matrices:\ $\Phi^{uv}$ will be associated with the oriented edge between $u$ and $v$ as follows.
    \begin{itemize}
        \item $f(uv)=\Phi^{uv}$, if we have added the arc $uv$;
        \item $f(uv)=(\Phi^{uv})^T$, if we have added the arc $vu$.
    \end{itemize}
To reduce from \mwdp to potential maximization, let the matrix $M$ on arc $uv$ with weight $c$ be given by the potential matrix above. Such a matrix is the potential of a $2\times2$ game given by $\left[\begin{array}{cc}
\frac{\phi_{11}}{2}, \frac{\phi_{11}}{2} & 0, \phi_{12} - \frac{\phi_{11}}{2} \\
\phi_{21} - \frac{\phi_{11}}{2}, 0 & \phi_{22}-\phi_{21}, \phi_{22} - \phi_{12} \\
\end{array}\right]$.
{As before, we have reduced in both directions and so we have a one-to-one translation between partitions and strategy profiles.
So, if we denote $P(\sbf)$ the partition associated with strategy profile $\sbf$ we get the following lemma, where $\Phi(\sbf)$ is the sum of the pairwise potentials.
\begin{lemma}
\label{lem:pot-to-mwdp}
For every possible strategy profile $\sbf$ of a pairwise-potential polymatrix game, it holds that $\Phi(\sbf) = w^{P(\sbf)}(D)$.
\end{lemma}

Lemma~\ref{lem:pot-to-mwdp} combined with Theorem~\ref{thm:dichotomy}, yield the following theorem. Again, to the best of our knowledge, this is the first dichotomy of this kind.

\begin{theorem}
\label{thm:pot-general}
Consider a binary-action potential polymatrix game on input graph $G$ with pairwise-potential matrices $\Phi^{uv}$, for every $uv \in E(G)$.
Finding a strategy profile that maximizes the potential function $\Phi$ can be solved in polynomial time if one of the following holds.
\begin{itemize}
   \item $\phi^{uv}_{11} + \phi^{uv}_{22} \geq \phi^{uv}_{12} + \phi^{uv}_{21}$ for every $uv \in E(G)$.
   \item $\phi^{uv}_{11} \geq \max \{\phi^{uv}_{12}, \phi^{uv}_{21}, \phi^{uv}_{22} \}$ for every $uv \in E(G)$.
   \item $\phi^{uv}_{22} \geq \max \{\phi^{uv}_{11}, \phi^{uv}_{12}, \phi^{uv}_{21} \}$ for every $uv \in E(G)$.
\end{itemize}
In every other case, the problem is \NP-hard.
\end{theorem}

Theorem~\ref{thm:pot-general} reveals another difference between pure-coordination and anti-coordination polymatrix games, when both are pairwise-potential games. 
While for the latter class of games it is \NP-hard to maximize the potential function (this was implied in \citeauthor{CD11-Cai-Daskalakis}~\shortcite{CD11-Cai-Daskalakis}), for the former class of games the problem can be solved in polynomial time.

\begin{corollary}
\label{cor:pot-pure-coord}
Maximizing the potential of a pairwise-potential binary-action pure-coordination polymatrix game can be solved in polynomial time.
\end{corollary}

\begin{proof}
We will show that when the game is pure-coordination, for every $uv \in E(G)$ the potential-matrix $\Phi^{uv}$ satisfies that $\phi^{uv}_{11}+\phi^{uv}_{22} \geq \phi^{uv}_{12}+\phi^{uv}_{21}$. 
Recall, since the game is pure-coordination the following hold. 
\begin{itemize}
    \item Outcome $(\one,\one)$, that corresponds to $\phi^{uv}_{11}$ of the potential-matrix, is a Nash equilibrium for the two-player game played on $uv$. So, since this is a Nash equilibrium {\em and} the two-player game is potential, any deviation from $(\one, \one)$ will lower the payoffs of {\em both} players, and thus it will lower the value of the pairwise-potential.
    Therefore, $\phi^{uv}_{11} \ge \max\{\phi^{uv}_{12},\phi^{uv}_{21}\}$.
    \item Outcome $(\two,\two)$, that corresponds to $\phi^{uv}_{22}$ of the potential-matrix, is a Nash equilibrium for the two-player game played on $uv$.
    Using verbatim the arguments as before, it follows that $\phi^{uv}_{22} \ge \max\{\phi^{uv}_{12},\phi^{uv}_{21}\}$.
\end{itemize}
 Our claim follows by combining the two bullets above.
%
\end{proof}
\section{Welfare-optimal Nash equilibria}
\label{sec:best-NE}
In this section we focus on the computation of welfare-optimal Nash equilibria in binary-action polymatrix games.
We focus only on pure-coordination games since the result of~\cite{CD11-Cai-Daskalakis} already implies that in anti-coordination games, welfare-optimal Nash equilibria are \NP-hard to compute. 
On the other hand, our results indicate that pure-coordination games are considerably more well-behaved, since they admit polynomial-time algorithms for welfare and potential maximization. 
Unfortunately, as we prove, when it comes to computing a ``best'' Nash equilibrium, things become almost immediately intractable even for the arguably simplest class of pure-coordination games, which we term {\em 2-type threshold-games}, which is a special case of {\em anonymous games}. 

\paragraph{\bf Anonymous pure-coordination games.} 
In an {\em anonymous game}~\cite{blonski1999anonymous}, the payoff of a player does not depend on the identity of their opponents, but only from the number of players that choose a specific action.
In an anonymous polymatrix game on graph $G$ every player $u \in V(G)$ is associated with a {\em single} payoff matrix $\Pi^u$ that is used in every two-player game they participate, i.e., player $u$ participates in the games $(\Pi^u, \Pi^v)$, for every $v \in E(G)$.
We say that an anonymous polymatrix game has $k$ types of players if there exists a set of payoff matrices \acal of size $k$ and $\Pi^u \in \acal$ for every $u \in V(G)$.
Finally, a pure-coordination game is $k$-type anonymous, if it satisfies the conditions above.

Next we focus on a special class of binary-action pure-coordination anonymous games, termed {\em threshold games} \cite{NearyNewton:2017:JMID};
their name follows from the fact that they model game-theoretically the classic threshold model of \citeauthor{Granovetter:1978:AJS}~\shortcite{Granovetter:1978:AJS}, wherein it is optimal for each player to choose each action once the fraction of their neighbors choosing that action exceeds a specific threshold.

\paragraph{\bf Threshold games.} A threshold-game is a binary-action anonymous pure-coordination polymatrix game, where every player $u$ is associated with a parameter $\gamma_u \in [0,1]$ and their payoff-matrix is 
$\Pi^{u} := \left[\begin{array}{cc}
\gamma_u & 0 \\
0 & 1-\gamma_u \\
\end{array}\right]$.

Clearly, finding optimal outcomes in a 1-type threshold-game (i.e., one wherein $\gamma_u$ is the same for all players $u$) is trivial.
However, as we now show, the problem becomes interesting when a second type is introduced. 
Specifically, we use the Language Game \cite{Neary:2012:GEB}, in which the population is partitioned into two types, $A$ and $B$, such that all players of type $A$ have threshold $\gamma_A$ and all players in Group $B$ have threshold $\gamma_B$, with $0\le \gamma_B \le \gamma_A \le 1$.
When $\gamma_B \le \frac{1}{2} \le \gamma_A$ a tension emerges:\ type-$A$ players prefer that everyone coordinates on action $\one$ while type-$B$ players prefer the opposite.
Even for this simple two-type threshold model, the following shows that most cases are intractable.


\begin{theorem}
\label{thm:max_NE-hard}
The complexity of finding a welfare-maxi\-mi\-zing Nash equilibrium in a threshold-game with two thresholds, $\gamma_A, \gamma_B$, such that $0 \le \gamma_B \le \gamma_A \le 1$, is as follows.
\begin{description}
\item[1.] If $\gamma_A \leq 1/2$ the problem is polynomial time.
\item[2.] If $\gamma_B \geq 1/2$ the problem is polynomial time.
\item[3.] If $\gamma_B = 0$ and $\gamma_A = 1$ the problem is polynomial time.
\item[4.] In all other cases the problem is NP-hard.
\end{description} 
\end{theorem}

\begin{proof}[Proof sketch]
{\bf Cases 1-2.}
If $\gamma_A \leq 1/2$, or $\gamma_B \geq 1/2$, then all players prefer the same action and selecting that action for each player 
gives a maximum-welfare Nash equilibrium.

{\bf Case 3.}
Now, let $\gamma_A=1$ and let $\gamma_B=0$ and $G$ be a graph and let $(A,B)$ be a partition of $V(G)$, so that if $u \in Y \Leftrightarrow \gamma_u = \gamma_Y$, where $Y \in \{ A,B\}$. 
We want to find a Nash equilibrium,  $(X_{\one},X_{\two})$, of maximum welfare, where $X_i$ denotes the set of players that choose action $i \in \{\one, \two\}$. 

Let  $(X_{\one},X_{\two})$ be a Nash equilibrium of $G$. This implies that all vertices in $X_{\two} \cap A$ only have edges to vertices in $X_{\two}$ (as
otherwise it is not a Nash equilibrium). Note that this means that any connected component $C$ in $G[A]$ is either subset of $X_{\one}$ or a subset of $X_{\two}$, else one can find an edge from a vertex in $X_{\two} \cap A$ to a vertex in $X_{\one}$. 
Analogously,
every connected component of $G[B]$ is either a subset of $X_{\one}$ or a subset of $X_{\two}$. Finally, given a connected component $C_A$ of $G[A]$ and a connected component $C_B$ of $G[B]$, if there is an edge in $G$ between a vertex $u\in C_A$ and a vertex $v\in C_B$, then it is not possible that $C_B\subseteq X_{\one}$ and $C_A\subseteq X_{\two}$.  

We can view the value of social welfare in a Nash equilibrium in term of the loss w.r.t. the sum of maximum achievable utility of each player which is in this case the sum of degrees of vertices in $G$ or $2|E(G)|$. The equilibrium that maximizes the social welfare then minimizes this loss. 
We now build an auxiliary digraph $D$ whose vertices are the components of $G[A]$, the components of $G[B]$ and two new vertices $s$ and $t$. 
There is an arc from $s$ to every component in $G[A]$, an arc from every component in $G[B]$ to $t$, and two opposite arcs between every connected component $C_A$ of $G[A]$ and every connected component $C_B$ of $G[B]$ such there is an edge in $G$ between a vertex $u\in C_A$ and a vertex $v\in C_B$.
The goal is to set up the weights of the arcs such that there is a one-to-one correspondence between minimal $s$-$t$ cuts in $D$ and Nash equilibria in $G$ given by: (i) every player in the ``$s$ side'' of the cut chooses action $\one$ and (ii) every player in the ``$t$ side'' chooses action $\two$. 
Moreover, the weight of the cut should be exactly the loss w.r.t. $2|E(G)|$. It is rather straightforward to verify that setting the weights as follows achieves this:
    every arc from $s$ to a component $C_A$ of $G[A]$ has weight sum of degrees of vertices in $C_A$; 
    every arc from a component $C_B$ of $G[B]$ to $t$ has weight sum of degrees of vertices in $C_B$; 
    every arc from a component $C_A$ of $G[A]$ to a component $C_B$ of $G[B]$ has weight two times the number of edges between $C_A$ and $C_B$;
    every arc from a component $C_B$ of $G[B]$ to a component $C_A$ of $G[A]$ has weight infinity (as players in $C_B$ cannot use strategy $\one$ if players in $C_A$ use strategy $\two$). 
To find a welfare-maximizing Nash equilibrium, it now suffices to find a minimum weight $s$-$t$ cut in $D$ that can be done in polynomial time~\cite{korte2011combinatorial}.

{\bf Case 4.} First note that if $\gamma_B = 0$, then $1/2 < \gamma_A < 1$, and we can exchange the role of $A$ and $B$ and that of the strategy $\one$ and strategy $\two$.
Hence, it suffices to prove the statement for $0 < \gamma_B <1/2 < \gamma_A \le 1$.

We will show the \NP-hardness result via a reduction from \textsc{Minimum Traversal} problem in $3$-uniform hypergraphs\footnote{The problem is 
 also called \textsc{$3$-Hitting Set}.}, which is known to be \NP-hard~\cite{garey1979computers}. 
 A $3$-uniform hypergraph $H$ has a set of vertices $V(H)$ and a set of hyperedges $E(H)$, where every hyperedge $e\in E(H)$ is a subset of vertices of size exactly $3$. 
 A traversal of $H$ is a set of vertices $X$ such that $X\cap e\neq \emptyset$ for every edge $e\in E(H)$. \textsc{Minimum Traversal} problem asks to find a traversal of the minimum size.

Let $H$ be a $3$-uniform hypergraph. We will create a graph, $G$, and a partition $(A,B)$ of $V(G)$ such that a solution to our problem for $G$ will give us a minimum traversal in $H$. We also refer the reader to Figure~\ref{fig:G} for an illustration.

\begin{figure}[t]
  \begin{center}
\tikzstyle{vertexB}=[circle,draw, minimum size=7pt, scale=0.9, inner sep=0.9pt]
\tikzstyle{vertexA}=[rectangle,draw, minimum size=6pt, scale=0.9, inner sep=0.9pt]
\tikzstyle{vertexBs}=[circle,draw, minimum size=6pt, scale=0.5, inner sep=0.9pt]
\begin{tikzpicture}[scale=0.3]


  \draw (0,12) rectangle (8,14); \node at (4,13) {$C_A$};
  \node (a1) at (1,13.2) [vertexA]{};
  \node (a2) at (2,12.8) [vertexA]{};
  \node (a3) at (6,13.2) [vertexA]{};
  \node (a4) at (7,12.8) [vertexA]{};

  \draw (10,12) rectangle (18,14); \node at (14,13) {$C_B$};
  \node (a1) at (11,13.2) [vertexB]{};
  \node (a2) at (12,12.8) [vertexB]{};
  \node (a3) at (16,13.2) [vertexB]{};
  \node (a4) at (17,12.8) [vertexB]{};

  \node (r1) at (1,9) [vertexB]{$r_{e_1}$};
  \node (r2) at (5,9) [vertexB]{$r_{e_2}$};
  \node (r3) at (9,9) [vertexB]{$r_{e_3}$};
  \node at (13,9) {$\cdots$};
  \node (r4) at (17,9) [vertexB]{$r_{e_m}$};

  \draw[line width=0.03cm] (r1) to (2,12);
  \draw[line width=0.03cm] (r1) to (10.5,12);

  \draw[line width=0.03cm] (r2) to (3,12);
  \draw[line width=0.03cm] (r2) to (11.5,12);

  \draw[line width=0.03cm] (r3) to (4,12);
  \draw[line width=0.03cm] (r3) to (12.5,12);

  \draw[line width=0.03cm] (r4) to (6,12);
  \draw[line width=0.03cm] (r4) to (14.5,12);

  \node (u1) at (1,4) [vertexB]{$u_1'$};
  \node (u2) at (4,4) [vertexB]{$u_2'$};
  \node (u3) at (7,4) [vertexB]{$u_3'$};
  \node (u4) at (10,4) [vertexB]{$u_4'$};
  \node at (13.5,4) {$\cdots$};
  \node (u5) at (17,4) [vertexB]{$u_n'$};

  \draw[line width=0.03cm] (r1) to (u1);
  \draw[line width=0.03cm] (r1) to (u2);
  \draw[line width=0.03cm] (r1) to (u3);

  \draw[line width=0.03cm] (r2) to (u1);
  \draw[line width=0.03cm] (r2) to (u3);
  \draw[line width=0.03cm] (r2) to (u4);

  \draw[line width=0.03cm] (r3) to (8.4,7.8);
  \draw[line width=0.03cm] (r3) to (9,7.8);
  \draw[line width=0.03cm] (r3) to (9.6,7.8);

  \draw[line width=0.03cm] (r4) to (16.4,7.8);
  \draw[line width=0.03cm] (r4) to (17,7.8);
  \draw[line width=0.03cm] (r4) to (17.6,7.8);

  \draw (0,-1.5) rectangle (2,1); \node at (1,-0.7) {$Z_{u_1'}$};
  \node (z11) at (0.4,0.5) [vertexBs]{};
  \node (z12) at (1,0.5) [vertexBs]{};
  \node (z13) at (1.6,0.5) [vertexBs]{};
  \draw[line width=0.08cm] (u1) to (1,1);

  \draw (3,-1.5) rectangle (5,1); \node at (4,-0.7) {$Z_{u_2'}$};
  \node (z11) at (3.4,0.5) [vertexBs]{};
  \node (z12) at (4,0.5) [vertexBs]{};
  \node (z13) at (4.6,0.5) [vertexBs]{};
  \draw[line width=0.08cm] (u2) to (4,1);

  \draw (6,-1.5) rectangle (8,1); \node at (7,-0.7) {$Z_{u_3'}$};
  \node (z11) at (6.4,0.5) [vertexBs]{};
  \node (z12) at (7,0.5) [vertexBs]{};
  \node (z13) at (7.6,0.5) [vertexBs]{};
  \draw[line width=0.08cm] (u3) to (7,1);

  \draw (9,-1.5) rectangle (11,1); \node at (10,-0.7) {$Z_{u_4'}$};
  \node (z11) at (9.4,0.5) [vertexBs]{};
  \node (z12) at (10,0.5) [vertexBs]{};
  \node (z13) at (10.6,0.5) [vertexBs]{};
  \draw[line width=0.08cm] (u4) to (10,1);

  \draw (16,-1.5) rectangle (18,1); \node at (17,-0.7) {$Z_{u_n'}$};
  \node (z11) at (16.4,0.5) [vertexBs]{};
  \node (z12) at (17,0.5) [vertexBs]{};
  \node (z13) at (17.6,0.5) [vertexBs]{};
  \draw[line width=0.08cm] (u5) to (17,1);

  \end{tikzpicture}
\caption{The construction of Case 4 in Theorem~\ref{thm:max_NE-hard}. The graph $G$ when $H$ is a $3$-uniform hypergraph with $m$ edges, including the edges $e_1=\{u_1,u_2,u_3\}$ and $e_2=\{u_1,u_3,u_4\}$.
The square vertices in $C_A$ denotes A-vertices and round vertices (everywhere else) denote B-vertices. 
} \label{fig:G}
\end{center} \end{figure}
For the reduction, we will need to fix some sizes that depend of $\gamma_A$, $\gamma_B$, $|V(H)|$, and $|E(H)|$. We postpone the selection of the sizes to the appendix, where we give exact proofs of the statements below. The graph $G$ consists of:
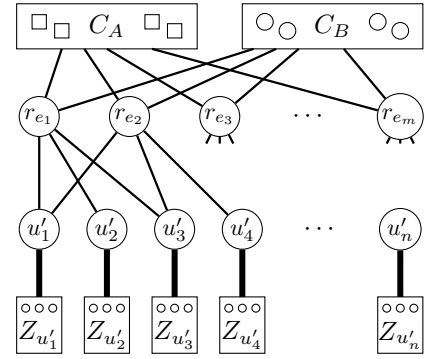
\begin{itemize}
    \item Cliques $C_A$, $C_B$ that are ``sufficiently large''; the sizes are selected in a way such that in any welfare-optimal Nash equilibrium, players in $C_A$ choose action $\one$ and in $C_B$ choose action $\two$. $A$ contains precisely the players in $C_A$, all the remaining players are in $B$. 
    \item Vertex sets $R = \bigcup_{e\in E(H)}\{r_e\}$, $V' = \bigcup_{u\in V(H)}\{u'\}$. There is an edge between $r_e$ and $u'$ iff $u\in e$. Moreover, there are $c_A$ edges from every $r_e$ to $C_A$ and $c_B$ edges from every $r_e$ to $C_B$. $c_A$ and $c_B$ are selected so that if every player in $C_A$ uses $\one$ and in $C_B$ $\two$, then $r_e$ (recall that $r_e\in B$) prefers $\one$ if at least one of its exactly three neighbors in $V'$ chooses $\one$.
    \item For every $u\in V(H)$, vertex set $Z_u$ of size $z$ that is to be fixed. There is an edge from every vertex $z\in Z_u$ to the vertex $u'\in V'$. Hence, in any Nash equilibrium $z$ chooses the same strategy as $u'$. This also allows $u'$ the freedom of choosing $\one$ even though $u'\in B$.
\end{itemize}

By a clever selection of the sizes, we can achieve that if $\sbf$ is a welfare-maximizing Nash equilibrium that partitions players into $(X_\one, X_\two)$ by the strategy they are using, then: 
\begin{enumerate}
    \item $V(C_A)\cup R\subseteq X_\one$, $V(C_B)\subseteq X_\two$. Intuitively, this is because the choice of $c_A$ and $c_B$ guarantees that the loss of utility for the players in $C_A$ when a player $r_e\in R$ selects $\two$, is larger than the gain players in $V(G)\setminus V(C_A)$ can get from this selection. 
    \item For every $e\in E(H)$, $r_e$ has a neighbor in $V'\cap X_\one$;
\end{enumerate}
Note that the second condition states that $V'\cap X_\one$ is a traversal of $H$. 
We now observe that because $V'\cup \bigcup_{u\in V(H)}Z_u\subseteq B$, there is a significant gain in overall welfare if we decrease the number of players in $V'$ that choose action $\one$, as long as we preserve the above two conditions. Therefore, a welfare-maximizing Nash equilibrium not only gives a traversal, but a minimum traversal~in~$H$.
\end{proof}




\section{Discussion}

Our paper provides several novel dichotomy results for the complexity of economically-efficient outcomes in general binary-action polymatrix games, coordination games, potential games, and threshold games. 

Our main tool for deriving the majority of these results is the dichotomy for \mwdp, a novel graph-theoretic problem, which we strongly believe will find applications in other domains too. 
To this end, we have already identified several problems arising from  graph theory whose complexity can easily be determined by our dichotomy for \mwdp. 
Our list includes both well-known problems ((Directed) Max Cut, (Directed) Min $(s,t)$-cut,  Max Density Subgraph) and new ones that we describe in the supplementary material.

\paragraph*{Acknowledgements.}
Anders Yeo's research was supported by the Danish research council for independent research under grant number DFF 7014-00037B.

\bibliographystyle{named}
\bibliography{ijcai23}

\appendix

\newpage
\onecolumn
\section{Proving Lemma \ref{lem:case-1} and Other Basic Results}

Recall that we consider the following general problem called \maxweightpartition (\mwdp).

\begin{description}
 \item[\mwdp(${\cal F}$):] We are given a family of $2 \times 2$ matrices ${\cal F}$. For each matrix $M \in {\cal F}$ let $m_{11}$, $m_{12}$, $m_{21}$ and 
$m_{22}$ denote the entries of $M$, such that 

\[
M = \left[\begin{array}{cc}
m_{11} & m_{12} \\
m_{21} & m_{22} \\
\end{array}\right]
\]

  An instance of {\it \mwdp(${\cal F}$)} consists of an oriented graph $D$ (that is, a directed graph without directed $2$-cycles),  
arc-weights $c: A(D) \rightarrow \mathbb{R}^+$ and 
an assignment, $f: A(D) \rightarrow {\cal F}$, of a matrix from ${\cal F}$ to each arc in $D$.

Given a partition, $P=(X_1,X_2)$, of $V(D)$ the weight of en arc $uv \in E(D)$ is defined as follows, where $M=f(uv)$ is the matrix
assigned to the arc $uv$.

 \[
w^P(uv) = \left\{
\begin{array}{rcl}
c(uv) \cdot m_{11} & \mbox{if} & u,v \in X_1 \\
c(uv) \cdot m_{22} & \mbox{if} & u,v \in X_2 \\
c(uv) \cdot m_{12} & \mbox{if} & u \in X_1 \mbox { and } v \in X_2 \\
c(uv) \cdot m_{21} & \mbox{if} & u \in X_2 \mbox { and } v \in X_1 \\
\end{array}
\right.
\]

  The weight of $D$, $w^P(D)$, is defined as the sum of all the weights of the arcs. That is, $w^P(D) = \sum_{a \in A(D)} w^P(a)$. 
  The aim is to find the partition, $P$, that maximizes $w^P(D)$.
\end{description}

\vspace{2mm}

{\bf Lemma} {\bf \ref{lem:case-1}}
{\em If $m_{11} + m_{22} \geq m_{12} + m_{21}$ for all matrices $M \in {\cal F}$ then {\it \mwdp(${\cal F}$)} can be solved in polynomial time.}

\pf
Let $(D,c,f)$ be an instance of {\it \mwdp(${\cal F}$)}.
We will construct a new (undirected) graph $H$ with vertex set $V(D) \cup \{s,t\}$ as follows. Let $UG(D)$ denote the undirected graph obtained from $D$ by removing
orientations of all arcs. 
Initially let $E(H) = E(UG(D)) \cup \{su,tu \; | \; u \in V(D) \}$ 
and let all edges in $H$ have weight zero. Now for each arc $uv \in A(D)$ we modify the weight function $w$ of $H$ as follows, where $M=f(uv)$ is the matrix associated
with the arc $uv$.

\begin{itemize}
\item Let $w(uv)=c(uv) \times \frac{m_{11} + m_{22} - m_{12} - m_{21}}{2}$. Note $w(uv)$ is the weight of the undirected edge $uv$ in $H$ associated with the arc $uv$ in $D$ and that this weight is non-negative.
\item Add $c(uv) \times \left( \frac{-m_{22}}{2}\right)$ to $w(su)$.
\item Add $c(uv) \times \left( \frac{-m_{22}}{2}\right)$ to $w(sv)$.
\item Add $c(uv) \times \frac{m_{21}-m_{11}-m_{12}}{2}$ to $w(tu)$.
\item Add $c(uv) \times \frac{m_{12}-m_{11}-m_{21}}{2}$ to $w(tv)$.
\end{itemize}

Let $\theta$ be the smallest possible weight of all edges in $H$ after we completed the above process ($\theta$ may be negative).
Now consider the weight function $w^*$ obtained from $w$ by 
subtracting $\theta$ from all edges incident with $s$ or $t$. That is, $w^*(uv)=w(uv)$ if $\{u,v\} \cap \{s,t\}=\emptyset$ and 
$w^*(uv)=w(uv) - \theta$ otherwise.  
Now we note that all $w^*$-weights in $H$ are non-negative.  We will show that for any $(s,t)$-cut, $(X_1,X_2)$ in $H$ (ie, $(X_1,X_2)$ partitions 
$V(H)$ and $s \in X_1$ and $t \in X_2$) the $w^*$-weight of the cut is equal to $-w^P(D) - |V(D)|\theta$, where $P$ is the partition 
$(X_1 \setminus \{s\},X_2 \setminus \{t\})$ in $D$.  Therefore a minimum weight cut in $H$ maximizes $w^P(D)$ in $D$.

Let $(X_1,X_2)$ be any $(s,t)$-cut in $H$.  For every $u \in V(D)$ we note that exactly one of the edges $su$ and $ut$ will belong to the cut.
Therefore, we note that the $w^*$-weight of the cut is $|V(D)|\theta$ less than the $w$-weight of the cut.
 It therefore suffices to show that the $w$-weight of the cut is $-w^P(D)$ (where  $P$ is the partition
$(X_1 \setminus \{s\},X_2 \setminus \{t\})$ of $V(D)$).

Let $uv \in A(D)$ be arbitrary and consider the following four possibilities.

\begin{itemize}
\item $u,v \in X_1$. In this case when we considered the arc $uv$ above (when defining the $w$-weights) we added  
$c(uv) \times \frac{m_{21}-m_{11}-m_{12}}{2}$ to $w(tu)$ and $c(uv) \times \frac{m_{12}-m_{11}-m_{21}}{2}$ to $w(tv)$.
So, we added the following amount to the $w$-weight of the  $(s,t)$-cut, $(X_1,X_2)$.
\[
c(uv) \times \left( \frac{m_{21}-m_{11}-m_{12}}{2} + \frac{m_{12}-m_{11}-m_{21}}{2} \right) = -c(uv) \cdot m_{11}
\]
\item $u,v \in X_2$. In this case when we considered the arc $uv$ above (when defining the $w$-weights) we added
$c(uv) \times \left(\frac{-m_{22}}{2}\right)$ to $w(su)$ and $c(uv) \times \left(\frac{-m_{22}}{2}\right)$ to $w(sv)$.
So, we added the following amount to the $w$-weight of the  $(s,t)$-cut, $(X_1,X_2)$.
\[
c(uv) \times \left( \frac{-m_{22}}{2} + \frac{-m_{22}}{2} \right) = -c(uv) \cdot m_{22}
\]
\item $u \in X_1$ and $v \in X_2$. In this case when we considered the arc $uv$ above (when defining the $w$-weights) we 
let $w(uv)=c(uv) \times \frac{m_{11} + m_{22} - m_{12} - m_{21}}{2}$ and added
$c(uv) \times \left( \frac{-m_{22}}{2}\right)$ to $w(sv)$ and
$c(uv) \times \frac{m_{21}-m_{11}-m_{12}}{2}$ to $w(tu)$.
So, we added the following amount to the $w$-weight of the  $(s,t)$-cut, $(X_1,X_2)$.
\[
c(uv) \times \left( \frac{m_{11} + m_{22} - m_{12} - m_{21}}{2} + \frac{-m_{22}}{2} + \frac{m_{21}-m_{11}-m_{12}}{2} \right) = -c(uv) \cdot m_{12}
\]
\item $v \in X_1$ and $u \in X_2$. In this case when we considered the arc $uv$ above (when defining the $w$-weights) we
let $w(uv)=c(uv) \times \frac{m_{11} + m_{22} - m_{12} - m_{21}}{2}$ and added
$c(uv) \times \left(\frac{-m_{22}}{2}\right)$ to $w(su)$ and
$c(uv) \times \frac{m_{12}-m_{11}-m_{21}}{2}$ to $w(tv)$.
So, we added the following amount to the $w$-weight of the  $(s,t)$-cut, $(X_1,X_2)$.
\[
c(uv) \times \left( \frac{m_{11} + m_{22} - m_{12} - m_{21}}{2} + \frac{-m_{22}}{2} + \frac{m_{12}-m_{11}-m_{21}}{2} \right) = -c(uv) \cdot m_{21}
\]
\end{itemize}

Therefore we note that in all cases we have added $-w^P(uv)$ to the $w$-weight of the  $(s,t)$-cut, $(X_1,X_2)$.
So the total $w$-weight of the  $(s,t)$-cut is $-w^P(D)$ as desired.

Analogously, if we have a partition $P=(X_1,X_2)$ of $V(D)$, then adding $s$ to $X_1$ and $t$ to $X_2$ we obtain a $(s,t)$-cut
with $w$-weight $-w^P(D)$ of $H$. 
As we can find a minimum $w^*$-weight cut in $H$ in polynomial time we can find a  partition $P=(X_1,X_2)$ of $V(D)$ with 
minimum value of $-w^P(D)$, which corresponds to the maximum value of $w^P(G)$. Therefore, 
{\it \mwdp(${\cal F}$)} can be solved in polynomial time in this case.~\qed

\begin{proposition}\label{thm2}
If $m_{11} = \max\{ m_{11}, m_{22}, m_{12}, m_{21} \}$ for all matrices $M \in {\cal F}$ then {\it \mwdp(${\cal F}$)} can be solved in polynomial time.
\end{proposition}

\pf
In this case the maximum is clearly obtained by the partition $P=(X_1,X_2)=(V(D),\emptyset)$ as the maximum possible value of $w^P(a)$ for every $a \in A(D)$
is obtained when both endpoints of $a$ belong to $X_1$.  Therefore, the problem is polynomial time solvable.~\qed


\begin{proposition}\label{thm3}
If $m_{22} = \max\{ m_{11}, m_{22}, m_{12}, m_{21} \}$ for all matrices $M \in {\cal F}$ then {\it \mwdp(${\cal F}$)} can be solved in polynomial time.
\end{proposition}

\pf
In this case the maximum is clearly obtained by the partition $P=(X_1,X_2)=(\emptyset,V(D))$ as the maximum possible value of $w^P(a)$ for every $a \in A(D)$
is obtained when both endpoints of $a$ belong to $X_2$.  Therefore, the problem is polynomial time solvable.~\qed

\2

A hypergraph, $H=(V,E)$, is $3$-{\em uniform} if all edges of $H$ contain three vertices. 
Two edges {\em overlap} if they have at least two vertices in common.
A hypergraph is {\em linear} if it contains no overlapping edges.
A hypergraph is {\em 2-colorable} if the vertex set can be colored with two colors such that every edge contains vertices of both colors.
The following result is well-known.

\2

\begin{theorem}\label{thmL}\cite{Lovasz}\footnote{In fact,  \cite{Lovasz} proves the result for 4-uniform hypergraphs, but a simple gadget allows one to reduce it to 3-uniform hypergraphs. Also, note that the 2-colorability of 3-uniform hypergraphs problem is equivalent to the monotone NAE 3-SAT problem whose NP-completeness follows from  Schaefer's dichotomy theorem
\cite{Schaefer78}, where monotonicity means that no negations are allowed in the SAT formulas.} 
It is NP-hard to decide if a $3$-uniform hypergraph is 2-colorable.
\end{theorem}

\2

We will extend the above theorem to the following theorem, where we only consider linear hypergraphs.

\begin{theorem} \label{thmNPhard}
It is NP-hard to decide if a linear $3$-uniform hypergraph is 2-colorable.
\end{theorem}

\pf
  We will reduce from the problem to decide if a $3$-uniform hypergraph is 2-colorable, which is NP-hard by  Theorem \ref{thmL}.
So let $H$ be a $3$-uniform hypergraph. We will reduce $H$ to a linear $3$-uniform hypergraph, $H'$, such that $H$ is 2-colorable if and only if
$H'$ is 2-colorable. This will complete the proof.

We may assume that $H$ is not linear, as otherwise we just let $H'=H$ and we are done. Let $e_y=\{x_1,x_2,y\}$ and $e_z=\{x_1,x_2,z\}$ be two edges that intersect
in at least two vertices. 
Let $F$ be a copy of the Fano plane and let $e=\{f_1,f_2,f_3\}$ be any edge in $F$. It is known that $F$ is not 2-colorable but $F-e$ is 2-colorable and any 
2-coloring of $F-e$ assigns the same color to all vertices of $e$ (as $F$ is not 2-colorable). Now add $F-e$ to $H-x_1$ and for every edge in $H$ that contains
$x_1$ replace $x_1$ by some vertex in $\{f_1,f_2,f_3\}$. Furthermore do this such that $e_y$ and $e_z$ are given different vertices in $\{f_1,f_2,f_3\}$.
Note that the resulting hypergraph is 2-colorable if and only if $H$ is 2-colorable. Furthermore, the number of edges that overlap has decreased. 
So, we have to perform the above operation at most ${|E(H)| \choose 2} \leq |E(H)|^2$ times in order to get rid of all overlapping edges. So this reduction is polynomial
and results in a linear $3$-uniform hypergraph, $H'$, that is 2-colorable if and only if
$H$ 2-colorable.~\qed

\begin{theorem} \label{thm4}
If $m_{11} + m_{22} < m_{12} + m_{21}$ and $\max\{m_{11},m_{22}\} < \max\{ m_{12}, m_{21} \}$  for some matrix $M \in {\cal F}$ then {\it \mwdp(${\cal F}$)} is NP-hard to solve.
\end{theorem}
\pf
We will reduce from the problem of determining if a linear $3$-uniform hypergraph is bipartite, which is NP-hard by Theorem~\ref{thmNPhard}.
Let $H=(V,E)$ be a linear $3$-uniform hypergraph. We will now construct an instance of {\it \mwdp(${\cal F}$)}, such that a solution to this instance
will tell us if $H$ is bipartite or not.

We start by letting $V(D)=V(H)$ and for each edge $e \in E(H)$ we add a directed $3$-cycle $C_e=xyzx$ to $D$, where $V(e)=\{x,y,z\}$.
Furthermore for each arc $a$ in the directed $3$-cycle let $c(a)=1$ and $f(a)=M$ ($M$ is defined in the statement of the theorem).
As $H$ is linear we note that $D$ is an oriented graph (i.e., it has no directed $2$-cycles or parallel arcs).
We now consider the following cases.

\2

{\bf Case 1. $m_{11}=m_{22}$:} In this case we will show that the above construction of $(D,c,f)$ is sufficient.
If $P=(X_1,X_2)$ is a partition on $V(H)$ (and therefore also a partition of $V(D)$) and $e \in E(H)$ then the following holds.

\begin{itemize}
 \item If $|V(e) \cap X_1|=3$ then the arcs of $C_e$ contribute $3m_{11}$ to $w^P(D)$.
 \item If $|V(e) \cap X_1|=2$ then the arcs of $C_e$ contribute $m_{11}+m_{12}+m_{21}$ to $w^P(D)$.
 \item If $|V(e) \cap X_1|=1$ then the arcs of $C_e$ contribute $m_{22}+m_{12}+m_{21}$ to $w^P(D)$.
 \item If $|V(e) \cap X_1|=0$ then the arcs of $C_e$ contribute $3m_{22}$ to $w^P(D)$.
\end{itemize}

So, as $m_{11}=m_{22}$, this implies the following.

\begin{itemize}
 \item If $|V(e) \cap X_1| \in \{0,3\}$ then the arcs of $C_e$ contribute $3m_{11}$ to $w^P(D)$.
 \item If $|V(e) \cap X_1| \in \{1,2\}$ then the arcs of $C_e$ contribute $m_{11}+m_{12}+m_{21}$ to $w^P(D)$.
\end{itemize}

And as $m_{12}+m_{21} > m_{11}+m_{22} = 2m_{11}$ we note that the maximum value of $w^P(D)$ is $|V(H)| (m_{11}+m_{12}+m_{21})$ if and only if
$|V(e) \cap X_1| \in \{1,2\}$ for all $e \in E(H)$, which is equivalent to $H$ being bipartite (as $X_1$ are the vertices with one color and $X_2$
are the vertices with the other color). So in this case we have the desired reduction.

\2

{\bf Case 2. $m_{11} \not= m_{22}$:}  Without loss of generality we may assume that $m_{11} > m_{22}$.
We now construct an instance, $(G(x,y),c,f)$, of {\it \mwdp(${\cal F}$)} such that every optimal solution $(X_1,X_2)$ for $G(x,y)$ has $x \in X_1$ and $y \in X_2$.

If $m_{12} > m_{21}$ then $m_{12}$ is the unique largest value in $\{ m_{11}, m_{22}, m_{12}, m_{21} \}$. 
Let $G(x,y)$ have vertex set $\{x,y\}$ and contain a single arc $xy$ with $f(xy)=M$ and $c(xy)=1$. Then
$x \in X_1$ and $y \in X_2$ for all optimal solutions $(X_1,X_2)$ to $G(x,y)$. 

If $m_{21} > m_{12}$ then $m_{21}$ is the unique largest value in $\{ m_{11}, m_{22}, m_{12}, m_{21} \}$. 
Analogously to the above construction, let $G(x,y)$ have vertex set $\{x,y\}$ and contain a single arc $yx$ with $f(yx)=M$ and $c(xy)=1$. Then 
$x \in X_1$ and $y \in X_2$ for all optimal solutions $(X_1,X_2)$ to $G(x,y)$.

Finally consider the case when $m_{12} = m_{21}$. 
Let  $G(x,y)$ have vertex set $\{x,y,z\}$ and contain a directed 
$3$-cycle $xyzx$ with $f(xy)=f(yz)=f(zx)=M$ and $c(xy)=c(yz)=2$ and $c(zx)=1$. Then
$x \in X_1$ (and $z \in X_1$) and $y \in X_2$ for all optimal solutions $(X_1,X_2)$ to $G(x,y)$ (as $m_{12}>m_{11} > m_{22}$).

Let $m^* = \max\{m_{12},m_{21}\}$ and let  $\theta=\frac{m_{11}-m_{22}}{m^*-m_{11}}$. Note that $\theta > 0$ as $m_{11} > m_{22}$ and $m^* > m_{11}$.
We now add $G(x,y)$ to $D$ (defined beforw Case 1) and add arcs between $G(x,y)$ and $D$ as follows.

\begin{itemize}
 \item If $m_{12} > m_{21}$. In this case add all possible arcs, $a$, from $x$ to $V(D)$ and let $c(a)=0$ and $f(a)=M$ for each of these.
Now for each $e \in E(H)$ we add $\theta$ to the $c$-value of the three arcs from $x$ to $V(C_e)$. This completes the construction.

 \item If $m_{12} \leq m_{21}$. In this case add all possible arcs, $a$, from $V(D)$ to $x$ and let the $c(a)=0$ and $f(a)=M$ for each of these.
Now for each $e \in E(H)$ we add $\theta$ to the $c$-value of the three arcs from $V(C_e)$ to $x$. This completes the construction.
\end{itemize}

Let $D^*$ denote the resulting digraph.  Now multiply the $c$-values of all arcs in $G(x,y)$ by a large enough constant to force
$x \in X_1$ and $y \in X_2$ for all optimal solutions $(X_1,X_2)$ to $D^*$.
Let $P=(X_1,X_2)$ be any partition of $V(H)$.  The following now holds for each $e \in E(H)$.

\begin{description}
 \item[If $|V(e) \cap X_1|=3$:] Then the arcs of $C_e$ contribute $3m_{11}$ to $w^P(D^*)$ and the weights we added to the arcs between $x$ and $V(e)$ when 
considering $e$ contribute $3 \theta m_{11}$ to $w^P(D^*)$. So all-in-all $w^P(D^*)$ has increased by the following amount due to $e$,
\[
 s_3 = 3m_{11} + 3 \theta m_{11} 
\]
 \item[If $|V(e) \cap X_1|=2$:] Then the arcs of $C_e$ contribute $m_{11}+m_{12}+m_{21}$ to $w^P(D^*)$ and the weights we added to the arcs between $x$ and $V(e)$ when 
considering $e$ contribute $2 \theta m_{11} + \theta m^*$ to $w^P(D^*)$. So all-in-all $w^P(D^*)$ has increased by the following amount due to $e$,
\[
 s_2 = m_{11}+m_{12}+m_{21} + \theta(2 m_{11} + m^*)
\]
 \item[If $|V(e) \cap X_1|=1$:] Then  the arcs of $C_e$ contribute $m_{22}+m_{12}+m_{21}$ to $w^P(D^*)$ and the weights we added to the arcs between $x$ and $V(e)$ when
considering $e$ contribute $\theta m_{11} + 2 \theta m^*$ to $w^P(D^*)$. So all-in-all $w^P(D^*)$ has increased by the following amount due to $e$,
 \[
 s_1 = m_{22}+m_{12}+m_{21} + \theta(m_{11} + 2 m^*)
\]
\item[If $|V(e) \cap X_1|=0$:] Then  the arcs of $C_e$ contribute $3m_{22}$ to $w^P(D^*)$ and the weights we added to the arcs between $x$ and $V(e)$ when
considering $e$ contribute $3 \theta m^*$ to $w^P(D^*)$. So all-in-all $w^P(D^*)$ has increased by the following amount due to $e$,
 \[
 s_0 = 3m_{22} + 3 \theta m^*
\]
\end{description}

We will now show that $s_0=s_3 < s_2=s_1$, which implies that $H$ is bipartite if and only if the optimal solution, $P=(X_1,X_2)$, for $D^*$ has 
 $|V(e) \cap X_1| \in \{1,2\}$ for all $e \in E(H)$. So we have reduced  
 the problem of determining if a linear $3$-uniform hypergraph is 2-colorable (which is NP-hard by Theorem~\ref{thmNPhard}) to {\it \mwdp(${\cal F}$)}
as desired. So the following claims complete the proof.

\2

{\bf Claim 2.A:} $s_0=s_3$.

{\bf Proof of Claim 2.A:} Note that the following holds as $\theta=\frac{m_{11}-m_{22}}{m^*-m_{11}}$, which proves the claim.
\[
\begin{array}{rcl}
s_3-s_0 & = & (3m_{11} + 3 \theta m_{11})-(3m_{22} + 3 \theta m^*) \\
        & = & 3(m_{11}-m_{22}) + 3 \frac{m_{11}-m_{22}}{m^*-m_{11}} (m_{11} - m^*) \\
        & = & 0 \\
\end{array}
\]
\2

{\bf Claim 2.B:} $s_2=s_1$.

{\bf Proof of Claim 2.B:}  Note that the following holds as $\theta=\frac{m_{11}-m_{22}}{m^*-m_{11}}$, which proves the claim.
\[
s_2-s_1  =  m_{11} - m_{22} + \theta (m_{11} - m^*) = (m_{11} - m_{22})-(m_{11} - m_{22}) = 0
\]
\2

{\bf Claim 2.C:} $s_3 < s_2$.

{\bf Proof of Claim 2.C:} Note that the following holds as $\theta=\frac{m_{11}-m_{22}}{m^*-m_{11}}$, which proves the claim.
\[
\begin{array}{rcl}
s_2-s_3 & = & (m_{12}+m_{21}-2m_{11}) +  \theta (m^* - m_{11}) \\
        & = & (m_{12}+m_{21}-2m_{11}) +  (m_{11}-m_{22}) \\
        & = & (m_{12}+m_{21}) - (m_{11}+m_{22}) \\
        & > & 0 \\
\end{array}
\]
\mbox{ } \hfill \qed

Note that Lemma~\ref{lem:case-1}, Proposition~\ref{thm2}, Proposition~\ref{thm3} and Theorem~\ref{thm4} imply a dichotomy for {\it \mwdp(${\cal F}$)} when 
${\cal F}$ only contains a single matrix. In order to obtain a complete dichotomy for all ${\cal F}$ we need the results in the following section.

\section{Proving Theorem \ref{thm:dichotomy}}

Let $M$ be a $2 \times 2$ matrix. We now define the following three properties of $M$ which may, or may not, hold.

\begin{description}
\item[Property (a):] $m_{11} + m_{22} \geq m_{12} + m_{21}$.
\item[Property (b):] $m_{11} = \max\{ m_{11}, m_{22}, m_{12}, m_{21} \}$.
\item[Property (c):] $m_{22} = \max\{ m_{11}, m_{22}, m_{12}, m_{21} \}$.
\end{description}

Note that Lemma~\ref{lem:case-1}, Proposition~\ref{thm2}, Proposition~\ref{thm3} and Theorem~\ref{thm4} can be reformulates as follows.
\begin{corollary}\label{cor1}
{\it \mwdp(${\cal F}$)} can be solved in polynomial time if all matrices in ${\cal F}$ satisfy Property (a) or all matrices in ${\cal F}$ satisfy Property (b)
or all matrices in ${\cal F}$ satisfy Property (c).  {\it \mwdp(${\cal F}$)} is NP-hard if some matrix in ${\cal F}$ does not satisfy any of the properties (a), (b) or (c).
\end{corollary}


\begin{theorem} \label{thm5}
We consider {\it \mwdp(${\cal F}$)}, where $M$ and $R$ are distinct matrices in ${\cal F}$.
If $M$ satisfies Property (b), but not Property (a) or Property (c) and $R$ satisfy Property (c), but not Property (b) (it may, or may not, 
satisfies Property (a)) then {\it \mwdp(${\cal F}$)} is NP-hard to solve.

By symmetry, if $M$ satisfies Property (c), but not Property (a) or Property (b) and $R$ satisfies Property (b), but not Property (c) (it may, or may not,
satisfy Property (a)) then the problem is also NP-hard to solve.

\end{theorem}

\pf
Let $M$ and $R$ be defined as in the statement of the theorem.  That is, the following holds.

\begin{description}
\item[(a):] $m_{11} = \max\{ m_{11}, m_{22}, m_{12}, m_{21} \}$ (as Property (b) holds for $M$).
\item[(b):] $m_{11} + m_{22} < m_{12} + m_{21}$ (as Property (a) does not hold for $M$)
\item[(c):] $m_{22} < \max\{ m_{11}, m_{22}, m_{12}, m_{21} \}$ (as Property (c) does not hold for $M$)
\item[(d):] $r_{22} = \max\{ r_{11}, r_{22}, r_{12}, r_{21} \}$ (as Property (c) holds for $R$). 
\item[(e):] $r_{11} < \max\{ r_{11}, r_{22}, r_{12}, r_{21} \}$ (as Property (b) does not hold for $R$)
\item[(f):] $m_{11} > m_{22}$, as by (a) and (c) we have $m_{11} = \max\{ m_{11}, m_{22}, m_{12}, m_{21} \} > m_{22}$. 
\item[(g):] $r_{22} > r_{11}$, as by (d) and (e) we have $r_{22} = \max\{ r_{11}, r_{22}, r_{12}, r_{21} \} > r_{11}$.
\end{description}

We will first construct an instance  $(D(x,x',y,y'),c,f)$ of {\it \mwdp(${\cal F}$)} such that for all optimal solutions
$P=(X_1,X_2)$ of the instance we have $x,x' \in X_1$ and $y,y' \in X_2$. We consider the following cases.

\2

{\bf Case 1: $m_{11}=m_{12}=m_{21}$ is not true and $r_{22}=r_{12}=r_{21}$ is also not true.} In this case $V(D(x,x',y,y'))=\{x,x',x'',y,y',y''\}$ and
$A(D(x,x',y,y'))=\{xx',x'x'',x''x,yy',y'y'',y''y\}$. Furthermore $c(a)=1$ for all arcs $a \in A(D(x,x',y,y'))$ and $f(xx')=f(x'x'')=f(x''x)=M$ and
$f(yy')=f(y'y'')=f(y''y)=R$. The only optimal solution $P=(X_1,X_2)$ is now $X_1=\{x,x',x''\}$ and $X_2=\{y,y',y''\}$ with value $w^P(D(x,x',y,y'))=3m_{11}+3r_{22}$.

\2

{\bf Case 2: $m_{11}=m_{12}=m_{21}$ and $r_{22}=r_{12}=r_{21}$.}   In this case $V(D(x,x',y,y'))=\{x,x',x_2,x_3,y,y',y_2,y_3\}$ and
$A(D(x,x',y,y'))=\{xx_2,xx_3,x'x_2,x'x_3,x_2x_3,yy_2,yy_3,y'y_2,y'y_3, y_2y_3\}$. Furthermore $c(a)=1$ for all arcs $a \in A(D(x,x',y,y'))$ and 
$f(xx_2)=f(xx_3)=f(x'x_2)=f(x'x_3)=f(y_2y_3)=M$ and
$f(yy_2)=f(yy_3)=f(y'y_2)=f(y'y_3)=f(x_2x_3)=R$. 
We note that $X_1=\{x,x',y_2,y_3\}$ and $X_2=\{y,y',x_2,x_3\}$ is an optimal solution with value $w^P(D(x,x',y,y'))=m_{11}+4m_{12}+r_{22}+4r_{21}
=5m_{11}+5r_{22}$.

For the sake of contradiction assume that some other optimal solution $P=(X_1,X_2)$ has $x \in X_2$. If both $x_2$ and $x_3$ are in $X_1$ then 
$w^P(x_2x_3)=r_{11} < r_{22}$, contradicting the fact that $P$ was optimal.  So either $x_2 \in X_2$ or $x_3 \in X_2$. but in either
case we have an arc $a \in \{xx_2,xx_3\}$ with $w^P(a)=m_{22} < m_{11}$, contradicting the fact that $P$ was optimal. So we must have $x \in X_1$ for all
optimal solutions $P=(X_1,X_2)$.
Analogously we can show that we also must have $x' \in X_1$, $y \in X_2$ and $y' \in X_2$, which completes the proof of Case 2.

\2

{\bf Case 3: $m_{11}=m_{12}=m_{21}$ is not true but $r_{22}=r_{12}=r_{21}$ is true.}  In this case $V(D(x,x',y,y'))=\{x,x',y,y',z\}$ and
$A(D(x,x',y,y'))=\{xx',x'z,zx,yy',y'z,zy\}$. Furthermore $c(a)=1$ for all arcs $a \in A(D(x,x',y,y'))$ and $f(xx')=f(x'z)=f(zx)=M$ and
$f(yy')=f(y'z)=f(zy)=R$.  We note that $X_1=\{x,x',z\}$ and $X_2=\{y,y'\}$ is an optimal solution with value 
$w^P(D(x,x',y,y'))=3m_{11}+r_{22}+r_{12}+r_{21}=3m_{11}+3r_{22}$. It is furthermore not difficult to check that this is the only optimal solution.

\2

{\bf Case 4: $m_{11}=m_{12}=m_{21}$ but $r_{22}=r_{12}=r_{21}$ is not true.}  This can be proved analogously to Case~3.

\2

Now that we have constructed $D(x,x',y,y')$ we can give the NP-hardness proof. We will reduce from the (unweighted) {\sc MaxCut} problem.
Let $G$ be any graph where we wish to find a cut with the maximum number of edges.
 We initially let $V(D^*)=V(G)$. Let $\epsilon=\frac{m_{11}-m_{22}}{r_{22}-r_{11}}$ and note that $\epsilon>0$ by properties~(f) and (g).
For each edge $e \in E(G)$, we now perform the following modifications to $D^*$.

\begin{description}
\item[1.] If $e=uv$ in $G$ then add the arc $uv$ to $D^*$ (any orientation of $e$ can be added to $D^*$) and let $c(uv)=1$ and $f(uv)=M$.
\item[2.] Add a copy, $D_e(x_e,x_e',y_e,y_e')$, of $D(x,x',y,y')$ to $D^*$ and add the following arcs.
\begin{description}
 \item[2a.] Add the directed $4$-cycle $C_e = ux_evy_eu$ to $D^*$ and let $c(a)=\epsilon$ and $f(a)=R$ for all $a \in A(C_e)$.
 \item[2b.] Add the arcs $A_e=\{vx'_e,x'_eu,vy'_e,y'_eu\}$ to $D^*$ and let $c(a)=1/2$ and $f(a)=M$ for all $a \in A(A_e)$.
\end{description}
\item[3.] Multiply $c(a)$ by a large constant, $K$, for all arcs $a$ which are part of a gadget $D_e(x_e,x_e',y_e,y_e')$ such that 
$x_e,x_e' \in X_1$ and $y_e,y_e' \in X_2$ for all optimal solutions $P=(X_1,X_2)$.
\end{description}

Let $P=(X_1,X_2)$ be an optimal solution for the instance  $(D^*,c,f)$. Let $ e=uv \in E(G)$ be arbitrary, such that $uv$ was added 
to $D^*$ in Step~1 above. Now consider the following four cases.
 
\2

{\bf Case A: $u,v \in X_1$.} The arcs added to $D^*$ in Step~1 and Steps~2a and 2b above contribute the following to $w^P(D^*)$.
\[
 s_A = m_{11}+\frac{2m_{11}+m_{12}+m_{21}}{2}+\epsilon ( 2r_{11}+r_{12}+r_{21} ) 
\]
\2
{\bf Case B: $u,v \in X_2$.} The arcs added to $D^*$ in Step~1 and Steps~2a and 2b above contribute the following to $w^P(D^*)$.
\[
 s_B = m_{22}+\frac{2m_{22}+m_{12}+m_{21}}{2}+\epsilon ( 2r_{22}+r_{12}+r_{21} ) 
\]
\2
{\bf Case C: $u \in X_1$ and $v \in X_2$.} The arcs added to $D^*$ in Step~1 and Steps~2a and 2b above contribute the following to $w^P(D^*)$.
\[
 s_C = m_{12}+\frac{2m_{21}+m_{11}+m_{22}}{2}+\epsilon ( r_{11}+r_{22}+r_{12}+r_{21} ) 
\]
\2
{\bf Case D: $u \in X_1$ and $v \in X_2$.} The arcs added to $D^*$ in Step~1 and Steps~2a and 2b above contribute the following to $w^P(D^*)$.
\[
 s_D = m_{21}+\frac{2m_{12}+m_{11}+m_{22}}{2}+\epsilon ( r_{11}+r_{22}+r_{12}+r_{21} )
\]
\2

We immediately see that $s_C=s_D$ as they are both equal to $m_{21}+m_{12} + (m_{11}+m_{22})/2+\epsilon ( r_{11}+r_{22}+r_{12}+r_{21} )$.
We will now show that $s_A=s_B$, which follows from the following, as $\epsilon=\frac{m_{11}-m_{22}}{r_{22}-r_{11}}$.

\[
\begin{array}{crcl}
& m_{11}-m_{22}  
& = &   \epsilon ( r_{22} - r_{11}) \\
\Updownarrow & & & \\
& 2m_{11}  + 2 \epsilon r_{11} 
& = &  2m_{22}+2 \epsilon r_{22} \\
\Updownarrow & & & \\
& m_{11}+\frac{2m_{11}+m_{12}+m_{21}}{2}+\epsilon ( 2r_{11}+r_{12}+r_{21} ) 
& = &  m_{22}+\frac{2m_{22}+m_{12}+m_{21}}{2}+\epsilon ( 2r_{22}+r_{12}+r_{21} ) \\
\Updownarrow & & & \\
& s_A & = & s_B \\ 
\end{array}
\]

This implies that $s_A=s_B$. We also note that the following holds.

\[
\begin{array}{rclcl}
(s_C+s_D)-(s_A+s_B) & = & (m_{12}+m_{21})-(m_{11}+m_{22}) & > & 0 \\
\end{array}
\]

So for every edge $e \in E(G)$ belonging to the cut $(X_1,X_2)$ we add $s_C$ ($=s_D$) to the value $w^P(D^*)$ and 
for every edge  $e \in E(G)$ not belonging to the cut $(X_1,X_2)$ we add $s_A$ ($=s_B<s_C=s_D$). 
So an optimal solution to max-cut for $G$ is equivalent to an optimal solution $(X_1,X_2)$ of $(D^*,c,f)$.
This completes the proof of the first part of the theorem.

The second part follows immediately by symmetry.~\qed

\begin{theorem} \label{thm6}
We consider {\it $\mwdp({\fcal}$)}, where $M$ and $R$ are distinct matrices in ${\cal F}$.
If $M$ satisfies Property (b), but not Property (a) or Property (c) and $R$ satisfies Property (a), but not Property (b) or Property (c)
 then {\it $\mwdp({\cal F}$)} is NP-hard to solve.

By symmetry, if $M$ satisfies Property (c), but not Property (a) or Property (b) and $R$ satisfies Property (a), but not Property (b) or Property (c)
 then the problem is also NP-hard to solve.
\end{theorem}

\pf
Let $M$ and $R$ be defined as in the statement of the theorem.  That is, the following holds.

\begin{description}
\item[(a):] $m_{11} = \max\{ m_{11}, m_{22}, m_{12}, m_{21} \}$ (as Property (b) holds for $M$).
\item[(b):] $m_{11} + m_{22} < m_{12} + m_{21}$ (as Property (a) does not hold for $M$)
\item[(c):] $m_{22} < \max\{ m_{11}, m_{22}, m_{12}, m_{21} \}$ (as Property (c) does not hold for $M$)
\item[(d):] $r_{11} + r_{22} \geq r_{12} + r_{21}$ (as Property (a) holds for $R$).
\item[(e):] $r_{22} < \max\{ r_{11}, r_{22}, r_{12}, r_{21} \}$ (as Property (c) does not hold for $R$).
\item[(f):] $r_{11} < \max\{ r_{11}, r_{22}, r_{12}, r_{21} \}$ (as Property (b) does not hold for $R$)
\item[(g):] $m_{11} > m_{22}$, as by (a) and (c) we have $m_{11} = \max\{ m_{11}, m_{22}, m_{12}, m_{21} \} > m_{22}$.
\item[(h):] $r_{12} \not= r_{21}$, as if $r_{12} = r_{21}$, then $r_{12} = r_{21} = \max\{ r_{11}, r_{22}, r_{12}, r_{21} \}$, by (e) and (f),
which is impossible by (d) (and (e) and (f)).
\end{description}

We will first construct an instance $(D(x,x',y,y'),c,f)$ of {\it \mwdp(${\cal F}$)} such that for all optimal solutions
$P=(X_1,X_2)$ of the instance we have $x,x' \in X_1$ and $y,y' \in X_2$. We consider the following cases, which exhaust all
possiblities by statement~(h) above.

\2

{\bf Case 1: $r_{12}>r_{21}$.} By (e) and (f) we note that $r_{12}$ is the unique largest value in $\{r_{11}, r_{22}, r_{12}, r_{21}\}$. 
Let $V(D(x,x',y,y'))=\{x,x',y,y'\}$ and $A(D(x,x',y,y'))=\{xy,x'y'\}$. 
Furthermore $c(xy)=c(x'y')=1$ and $f(xy)=f(x'y')=R$.
The only optimal solution $P=(X_1,X_2)$ is now $X_1=\{x,x'\}$ and $X_2=\{y,y'\}$ with value $w^P(D(x,x',y,y'))=2r_{12}$.

\2

{\bf Case 2: $r_{21}>r_{12}$.} By (e) and (f) we note that $r_{21}$ is the unique largest value in $\{r_{11}, r_{22}, r_{12}, r_{21}\}$. 
Let $V(D(x,x',y,y'))=\{x,x',y,y'\}$ and $A(D(x,x',y,y'))=\{yx,y'x'\}$. 
Furthermore $c(yx)=c(y'x')=1$ and $f(yx)=f(y'x')=R$.
The only optimal solution $P=(X_1,X_2)$ is now $X_1=\{x,x'\}$ and $X_2=\{y,y'\}$ with value $w^P(D(x,x',y,y'))=2r_{21}$.

\2

Now that we have constructed $D(x,x',y,y')$ we can give the NP-hardness proof. We will reduce from the (unweighted) {\sc MaxCut} problem.
So let $G$ be any graph where we are interested in finding a cut with the maximum number of edges.
 We initially let $V(D^*)=V(G)$.  Let $r^* = \max \{r_{12},r_{21}\}$ and let $\epsilon=\frac{m_{11}-m_{22}}{r^*-r_{11}}$ and note that $\epsilon>0$ by 
properties~(g), (e) and (f).
For each edge $e \in E(G)$ we now perform the following modifications to $D^*$.

\begin{description}
\item[1.] If $e=uv$ in $G$ then add the arc $uv$ to $D^*$ (any orientation of $e$ can be added to $D^*$) and let $c(uv)=1$ and $f(uv)=M$.
\item[2.] Add a copy, $D_e(x_e,x_e',y_e,y_e')$, of $D(x,x',y,y')$ to $D^*$ and add the following arcs.
\begin{description}
 \item[2a.] If $r_{12}>r_{21}$ then add the arcs $A_e'=\{x_eu,x_ev\}$ to $D^*$ and let $c(x_eu)=c(x_ev)=\epsilon$ and $f(x_eu)=f(x_ev)=R$.

And if $r_{12}<r_{21}$ then add the arcs $A_e'=\{ux_e,vx_e\}$ to $D^*$ and let $c(u_ex)=c(vx_e)=\epsilon$ and $f(ux_e)=f(vx_e)=R$.

 \item[2b.] Add the arcs $A_e=\{vx'_e,x'_eu,vy'_e,y'_eu\}$ to $D^*$ and let $c(a)=1/2$ and $f(a)=M$ for all $a \in A(A_e)$.
\end{description}
\item[3.] Multiply $c(a)$ by a large constant, $K$, for all arcs $a$ which are part of a gadget $D_e(x_e,x_e',y_e,y_e')$ such that 
$x_e,x_e' \in X_1$ and $y_e,y_e' \in X_2$ for all optimal solutions $P=(X_1,X_2)$.
\end{description}

Let $P=(X_1,X_2)$ be an optimal solution for the instance  $(D^*,c,f)$. Let $e=uv \in E(G)$ be arbitrary, such that $uv$ was added 
to $D^*$ in Step~1 above. Now consider the following four cases.
 
\2

{\bf Case A: $u,v \in X_1$.} The arcs added to $D^*$ in Step~1 and Steps~2a and 2b above contribute the following to $w^P(D^*)$.
\[
 s_A = m_{11}+\frac{2m_{11}+m_{12}+m_{21}}{2}+\epsilon (2r_{11}) 
\]
\2
{\bf Case B: $u,v \in X_2$.} The arcs added to $D^*$ in Step~1 and Steps~2a and 2b above contribute the following to $w^P(D^*)$.
\[
 s_B = m_{22}+\frac{2m_{22}+m_{12}+m_{21}}{2}+\epsilon ( 2r^* ) 
\]
\2
{\bf Case C: $u \in X_1$ and $v \in X_2$.} The arcs added to $D^*$ in Step~1 and Steps~2a and 2b above contribute the following to $w^P(D^*)$.
\[
 s_C = m_{12}+\frac{2m_{21}+m_{11}+m_{22}}{2}+\epsilon ( r^*+r_{11} )
\]
\2
{\bf Case D: $u \in X_2$ and $v \in X_1$.} The arcs added to $D^*$ in Step~1 and Steps~2a and 2b above contribute the following to $w^P(D^*)$.
\[
 s_D = m_{21}+\frac{2m_{12}+m_{11}+m_{22}}{2}+\epsilon ( r^*+r_{11} )
\]
\2

We immediately see that $s_C=s_D$ as they are both equal to $m_{21}+m_{12} + (m_{11}+m_{22})/2+ \epsilon ( r^*+r_{11} )$.
We will now show that $s_A=s_B$, which follows from the following, as $\epsilon=\frac{m_{11}-m_{22}}{r^*-r_{11}}$.

\[
\begin{array}{crcl}
& m_{11}-m_{22}  
& = &   \epsilon ( r^* - r_{11}) \\
\Updownarrow & & & \\
& 2m_{11}  + 2 \epsilon r_{11} 
& = &  2m_{22}+2 \epsilon r^* \\
\Updownarrow & & & \\
& m_{11}+\frac{2m_{11}+m_{12}+m_{21}}{2}+ 2\epsilon r_{11}  
& = &  m_{22}+\frac{2m_{22}+m_{12}+m_{21}}{2}+2\epsilon  r^* \\
\Updownarrow & & & \\
& s_A & = & s_B \\ 
\end{array}
\]

This implies that $s_A=s_B$. We also note that the following holds.

\[
\begin{array}{rclcl}
(s_C+s_D)-(s_A+s_B) & = & (m_{12}+m_{21})-(m_{11}+m_{22}) & > & 0 \\
\end{array}
\]

So for every edge $e \in E(G)$ belonging to the cut $(X_1,X_2)$ we add $s_C$ ($=s_D$) to the value $w^P(D^*)$ and 
for every edge  $e \in E(G)$ not belonging to the cut $(X_1,X_2)$ we add $s_A$ ($=s_B<s_C=s_D$). 
So an optimal solution to max-cut for $G$ is equivalent to an optimal solution $(X_1,X_2)$ of $(D^*,c,f)$.
This completes the proof of the first part of the theorem.

The second part follows immediately by symmetry.~\qed

\2

Now we are ready to prove our dichotomy. 

\2

{\bf Theorem} {\bf \ref{thm:dichotomy}}
{\em {\it \mwdp(${\cal F}$)} can be solved in polynomial time if all matrices in ${\cal F}$ satisfy Property (a) or all matrices in ${\cal F}$ satisfy Property (b) or all matrices in ${\cal F}$ satisfy Property (c).
Otherwise, {\it \mwdp(${\cal F}$)}  is NP-hard.} 


\2

\pf
Parts (1), (2) and (3) follow from Corollary~\ref{cor1} (and also from Lemma~\ref{lem:case-1}, Proposition~\ref{thm2} and Proposition~\ref{thm3}).
So now assume that  some matrix in ${\cal F}$ does not satisfy Property (a), some matrix (possibly the same)
does not satisfy Property (b) and some matrix (possibly the same)
does not satisfy Property (c).
Let $M$ be a matrix in ${\cal F}$ that does not satisfy Property (a). If $M$ does not satisfy Property (b) and does not satisfy Property (c) then 
we are done by Theorem~\ref{thm4}, so, without loss of generality, we may assume that $M$ satisfies Property (b). If $M$ satisfies
Property (c), then $m_{11} = m_{22} = \max\{ m_{11}, m_{22}, m_{12}, m_{21} \}$, which implies that Property (a) holds, a contradiction.
So assume that $M$ satisfies Property (b), but not Property (a) or Property (c).

Let $R$ be a matrix in ${\cal F}$ that does not satisfy Property (b). 
If $R$ satisfies Property (c), then we are done by Theorem~\ref{thm5}, so we may assume that this is not the case.
If $R$ satisfies Property (a), then we are done by Theorem~\ref{thm6}, so we may assume that this is not the case.
But now $R$ does not satisfy Property (a), Property (b) or Property (c) and we are done by Theorem~\ref{thm4}.
\qed

\section{Proving Case 3 of Theorem \ref{thm:max_NE-hard}}
Let $\gamma_A=1$ and let $\gamma_B=0$ and $G$ be a graph and let $(A,B)$ be a partition of $V(G)$, so that if $u \in Y \Leftrightarrow \gamma_u = \gamma_Y$, where $Y \in \{ A,B\}$. 
We want to find a Nash equilibrium,  $(X_{\one},X_{\two})$, of maximum welfare, where $X_i$ denotes the set of players that choose action $i \in \{\one, \two\}$. 

Let  $(X_{\one},X_{\two})$ be a Nash equilibrium of $G$. This implies that all vertices in $X_{\two} \cap A$ only have edges to vertices in $X_{\two}$ (as
otherwise it is not a Nash equilibrium). Note that this means that any connected component $C$ in $G[A]$ is either subset of $X_{\one}$ or a subset of $X_{\two}$, else one can find an edge from a vertex in $X_{\two} \cap A$ to a vertex in $X_{\one}$. 
Analogously,
every connected component of $G[B]$ is either a subset of $X_{\one}$ or a subset of $X_{\two}$. Finally, given a connected component $C_A$ of $G[A]$ and a connected component $C_B$ of $G[B]$, if there is an edge in $G$ between a vertex $u\in C_A$ and a vertex $v\in C_B$, then it is not possible that $C_B\subseteq X_{\one}$ and $C_A\subseteq X_{\two}$.


We now build a digraph $D$ as follows. Let $C_1^A, C_2^A, \ldots, C_l^A$ be the components of $G[A]$ and 
$C_1^B, C_2^B, \ldots, C_m^B$ the components of $G[B]$. Let $V(D)=\{s,t\} \cup \{a_1,a_2,\ldots,a_l\} \cup \{b_1,b_2,\ldots,b_m\}$.
For each $C_i^A$ add an arc from $s$ to $a_i$ of weight $2|E(C_i^A)|$ and for each $C_i^B$ add an arc from $b_i$ to $t$ of weight $2|E(C_i^B)|$.
Now for each $i \in \{1,2,\ldots,l\}$ and $j \in \{1,2, \ldots, m\}$ we add an arc from $a_i$ to $b_j$ of weight $|E(C_i^A,C_j^B)|$ and
if $|E(C_i^A,C_j^B)|>0$ then add an arc from $b_j$ to $a_i$ of infinite weight. This completes the construction of $D$.

Let $(S,T)$ be a $(s,t)$-cut (i.e., $s \in S$ and $t \in T$) in $D$ containing no arc of infinite weight. Let the weight of $(S,T)$ be $w(S,T)$.
We will now construct a Nash equilibrium $(X_{\one},X_{\two})$ of $G$.
For every $a_i \in S$ add all vertices in $C_i^A$ to $X_{\one}$ and for every $a_i \in T$ add all vertices in $C_i^A$ to $X_{\two}$.
Analogously, for every $b_j \in S$ add all vertices in $C_j^B$ to $X_{\one}$ and for every $b_j \in T$ add all vertices in $C_j^B$ to $X_{\two}$.
This defines a partition $(X_{\one},X_{\two})$.

As no arc of infinite weight  exists in the cut, we note that that $(X_{\one},X_{\two})$ is a Nash equilibrium.
We will now show that the welfare of $(X_{\one},X_{\two})$ is $w(X_{\one},X_{\two})=2|E(G)|-|E(A,B)| - w(S,T)$.
If $sa_i$ belongs to the cut then all arcs in $C_i^A$ contribute zero to the welfare and zero to the formula for $w(X_{\one},X_{\two})$.
If $b_jt$ belongs to the cut then all arcs in $C_j^B$ contribute zero to the welfare and zero to the formula for $w(X_{\one},X_{\two})$.
If $a_ib_j$ belongs to the cut then all arcs in $E(C_i^A,C_j^B)$ contribute zero to the welfare and zero to the formula for $w(X_{\one},X_{\two})$.
All other arcs contribute two to the welfare if both end-points lie in $A$ or both end-points lie in $B$ and if exactly one
end-point lies in $A$ and the other in $B$ then the edge contributes one to the welfare.
So the welfare of $(X_{\one},X_{\two})$ is indeed $w(X_{\one},X_{\two})=2|E(G)|-|E(A,B)| - w(S,T)$.

Conversely let $(Y_{\one},Y_{\two})$ be a Nash equilibrium of maximum weight. We saw above that all vertices in $C_i^A$ belong to $Y_{\one}$ or to $Y_{\two}$.
If they belong to $Y_{\one}$, then let $a_i$ belong to $S$ and otherwise let $a_i$ belong to $T$. Analogously if all
vertices of $C_j^B$ belong to $Y_{\one}$ then let $b_j$ belong to $S$ and otherwise let $b_j$ belong to $T$.
Furthermore let $s$ belong to $S$ and $t$ belong to $T$, which gives us a $(s,t)$-cut $(S,T)$.
We will show that the weight of the cut, $w(S,T)$, satisfies $w(Y_{\one},Y_{\two})=2|E(G)|-|E(A,B)| - w(S,T)$, where $w(Y_{\one},Y_{\two})$ is the welfare of
$(Y_{\one},Y_{\two})$. As $(Y_{\one},Y_{\two})$ is exactly the Nash equilibrium we obtained above if we had started with $(S,T)$, we note that
$w(Y_{\one},Y_{\two})=2|E(G)|-|E(A,B)| - w(S,T)$ indeed holds.

Therefore, if the maximum welfare Nash equilibrium has welfare ${\rm wel}^{opt}$ and the minimum weight $(s,t)$-cut has weight
${\rm cut}^{opt}$ then ${\rm wel}^{opt}=2|E(G)|-|E(A,B)| - {\rm cut}^{opt}$.
As we can find a minimum $(s,t)$-cut in polynomial time \cite{korte2011combinatorial}, this gives us a polynomial time algorithm for finding a maximum welfare Nash equilibrium
when  $\gamma_A=1$ and $\gamma_B=0$. So we are done.

\section{Proving Case 4 of Theorem \ref{thm:max_NE-hard}}

Recall that it suffices to prove the theorem for $0 < \gamma_B <1/2 < \gamma_A \le 1$.
Thus it is enough to prove Theorem \ref{mainNPwelfareNash} below.

We will start with two useful lemmas.

\begin{lemma}\label{lemmaAclique}
Let  $\gamma_A$ be any real such that $ 1/2 < \gamma_A \leq 1$.
Let $G$ be the complete graph with $a\ge 2$ vertices of type A (and no vertices of type B).
The maximum welfare is obtained when all players choose action $\one$. Furthermore, if $(X_{\one},X_{\two}) \not= (V(G),\emptyset)$ then the welfare is at least 
$\min\{2\gamma_A(a-1), a(a-1)(2 \gamma_A-1)\}$ less than the optimum.
\end{lemma}

\pf
  Define $G$ as in the statement.  Assume that $(X_{\one},X_{\two})$ is a partition of $V(G)$. 
The welfare of this partition is denoted by $w(X_{\one},X_{\two})$. First assume that  $(X_{\one},X_{\two}) \not= (V(G),\emptyset)$.
Then the following holds, where $q=|X_{\two}|$,

\[
\begin{array}{rcl}
w(V(G),\emptyset) - w(X_{\one},X_{\two}) & = & 
{|X_{\two}| \choose 2}(2 \gamma_A) + |X_{\one}|\cdot |X_{\two}|(2\gamma_A) - {|X_{\two}| \choose 2}(2 (1-\gamma_A)) \\
& = & q(q-1)(\gamma_A-(1-\gamma_A)) + 2(a-q)q\gamma_A \\
& = & q\left( (q-1)(2\gamma_A-1) + 2a\gamma_A - 2q\gamma_A  \right)   \\
& = & q\left( 2q\gamma_A -q  -2 \gamma_A + 1 + 2a\gamma_A - 2q\gamma_A  \right)   \\
& = & q\left( 2\gamma_A(a-1) -q + 1  \right)   \\
\end{array}
\]

Consider the function $f(q)=q\left( 2\gamma_A(a-1) -q + 1  \right)$ and note that $f'(q)=-2q + 2\gamma_A(a-1)+1$, which implies that 
the minimum value of $f(q)$ in any interval is found in one of its endpoints. As $1 \leq q \leq a$ the minimum value of $f(q)$ is therefore
$f(1)$ or $f(a)$, which implies the following.

\[
w(V(G),\emptyset) - w(X_{\one},X_{\two})  \geq  \min\{f(1),f(a)\} = \min\{2\gamma_A(a-1), a(a-1)(2 \gamma_A-1)\}
\]

As the above minimum is positive, we note that if $(X_{\one},X_{\two}) \not= (V(G),\emptyset)$ then $(X_{\one},X_{\two})$ does not obtain the maximum welfare, so 
the maximum welfare is obtained when all players choose action $\one$. Furthermore, if  $(X_{\one},X_{\two}) \not= (V(G),\emptyset)$ then the welfare is at least
$\min\{2\gamma_A(a-1), a(a-1)(2 \gamma_A-1)\}$ less than the optimum as seen above.~\qed

\begin{corollary}\label{corBclique}
Let  $\gamma_B$ be any real such that $ 0 \leq \gamma_B < 1/2$.
Let $G$ be the complete graph with $b\ge 2$ vertices of type B (and no vertices of type A).
The maximum welfare is obtained when all players choose action $\two$ and if $(X_{\one},X_{\two}) \not= (\emptyset,V(G))$ then the welfare is at least
$\min\{2(1-\gamma_B)(b-1), b(b-1)(1-2\gamma_B)\}$ less than the optimum.
\end{corollary}

\pf
By replacing action $\two$ with $\one$ and vice versa and letting $\gamma_B^{new}=1-\gamma_B$ then we obtain an equivalent problem.
The following now completes the proof, by Lemma~\ref{lemmaAclique},
\[
 \min\{2\gamma_B^{new}(b-1), b(b-1)(2 \gamma_B^{new}-1)\} =  \min\{2(1-\gamma_B)(b-1), b(b-1)(1-2\gamma_B)\}
\]
\qed

\begin{lemma} \label{NashEq}
Let  $\gamma_A$ and $\gamma_B$ be given rationals, such that $0<\gamma_B \leq 1/2 \leq \gamma_A \leq 1$.
Let $G$ be any graph and let $(A,B)$ be a partition of $V(G)$.
The partition $(X_{\one},X_{\two})$ is a Nash equilibrium if and only if the following holds for all $u \in V(G)$.

\begin{itemize}
\item For every vertex $u \in X_{\one} \cap A$ we have $N(u) \cap X_{\two}=\emptyset$ or $\frac{|N(u) \cap X_{\one}|}{|N(u) \cap X_{\two}|} \geq \frac{1-\gamma_A}{\gamma_A}$.
\item For every vertex $u \in X_{\two} \cap A$ we have $N(u) \cap X_{\one}=\emptyset$ or $\frac{|N(u) \cap X_{\two}|}{|N(u) \cap X_{\one}|} \geq \frac{\gamma_A}{1-\gamma_A}$.  
\item For every vertex $u \in X_{\one} \cap B$ we have $N(u) \cap X_{\two}=\emptyset$ or $\frac{|N(u) \cap X_{\one}|}{|N(u) \cap X_{\two}|} \geq \frac{1-\gamma_B}{\gamma_B}$.
\item For every vertex $u \in X_{\two} \cap B$ we have $N(u) \cap X_{\one}=\emptyset$ or $\frac{|N(u) \cap X_{\two}|}{|N(u) \cap X_{\one}|} \geq \frac{\gamma_B}{1-\gamma_B}$.
\end{itemize}
\end{lemma}

\pf 
  Let $(X_{\one},X_{\two})$ be a partition of $V(G)$.
  Let $u \in X_{\one} \cap A$ be arbitrary. If we move $u$ from $X_{\one}$ to $X_{\two}$ then $u$'s gain is $|N(u) \cap X_{\two}|(1-\gamma_A) - |N(u) \cap X_{\one}|\gamma_A$.
This gain is greater than zero if $\frac{|N(u) \cap X_{\one}|}{|N(u) \cap X_{\two}|} < \frac{1-\gamma_A}{\gamma_A}$, so $u$ satisfies the Nash equilibrium property 
if and only if $N(u) \cap X_{\two}=\emptyset$ or $\frac{|N(u) \cap X_{\one}|}{|N(u) \cap X_{\two}|} \geq \frac{1-\gamma_A}{\gamma_A}$.

  Now let $u \in X_{\two} \cap A$ be arbitrary.  If we move $u$ from $X_{\two}$ to $X_{\one}$ then $u$'s gain is $|N(u) \cap X_{\one}|\gamma_A - |N(u) \cap X_{\two}|(1-\gamma_A)$.
This gain is greater than zero if $\frac{|N(u) \cap X_{\two}|}{|N(u) \cap X_{\one}|} < \frac{\gamma_A}{1-\gamma_A}$, so $u$ satisfies the Nash equilibrium property 
if and only if $N(u) \cap X_{\one} = \emptyset$ or $\frac{|N(u) \cap X_{\two}|}{|N(u) \cap X_{\one}|} \geq \frac{\gamma_A}{1-\gamma_A}$.
 
  Now let $u \in X_{\one} \cap B$ be arbitrary. If we move $u$ from $X_{\one}$ to $X_{\two}$ then $u$'s gain is $|N(u) \cap X_{\two}|(1-\gamma_B) - |N(u) \cap X_{\one}|\gamma_B$.
This gain is greater than zero if $\frac{|N(u) \cap X_{\one}|}{|N(u) \cap X_{\two}|} < \frac{1-\gamma_B}{\gamma_B}$, so $u$ satisfies the Nash equilibrium property
if and only if $N(u) \cap X_{\two}=\emptyset$ or $\frac{|N(u) \cap X_{\one}|}{|N(u) \cap X_{\two}|} \geq \frac{1-\gamma_B}{\gamma_B}$.

  Now let $u \in X_{\two} \cap B$ be arbitrary.  If we move $u$ from $X_{\two}$ to $X_{\one}$ then $u$'s gain is $|N(u) \cap X_{\one}|\gamma_B - |N(u) \cap X_{\two}|(1-\gamma_B)$.
This gain is greater than zero if $\frac{|N(u) \cap X_{\two}|}{|N(u) \cap X_{\one}|} < \frac{\gamma_B}{1-\gamma_B}$, so $u$ satisfies the Nash equilibrium property
if and only if $N(u) \cap X_{\one} = \emptyset$ or $\frac{|N(u) \cap X_{\two}|}{|N(u) \cap X_{\one}|} \geq \frac{\gamma_B}{1-\gamma_B}$.~\qed

\begin{theorem}\label{mainNPwelfareNash} Let  $\gamma_A$ and $\gamma_B$ be reals such that $0<\gamma_B < 1/2 < \gamma_A \leq 1$.
Then it is \NP-hard to find a maximum welfare Nash equilibrium.
\end{theorem}

\pf We will show the \NP-hardness result via a reduction from \textsc{Minimum Traversal} problem in $3$-uniform hypergraphs, which is known to be \NP-hard~\cite{garey1979computers}. 
 A $3$-uniform hypergraph $H$ has a set of vertices $V(H)$ and a set of hyperedges $E(H)$, where every edge $e\in E(H)$ is a subset of vertices of size exactly $3$. 
 A traversal of $H$ is a set of vertices $X$ such that $X\cap e\neq \emptyset$ for every edge $e\in E(H)$. \textsc{Minimum Traversal} problem asks to find a traversal of the minimum size.

Let  $H$ be a $3$-uniform hypergraph where we want to find a minimum transversal.
Let  $\gamma_A$ and $\gamma_B$ be given reals, such that $0<\gamma_B < 1/2 < \gamma_A \leq 1$.
We first define the following constants.

\begin{description}
\item[(a):] Let $\theta=6|E(H)| + 2|V(H)|\lceil 3|E(H)|\frac{1-\gamma_B}{\gamma_B(1-2\gamma_B)}\rceil$. 
\item[(b):] Let $x_B$ be the smallest possible positive integer such that $x_B(1-\gamma_B)(\gamma_A-\gamma_B)/\gamma_B > \theta$.
\item[(c):] Let $x_A$ be any integer such that the following holds (which is possible as $1/\gamma_B > 1$ (in fact, $1/\gamma_B > 2$).

\[
(x_B+3)\frac{1-\gamma_B}{\gamma_B} - \frac{1}{\gamma_B} < x_A < (x_B+3)\frac{1-\gamma_B}{\gamma_B}
\]

\item[(d):] Let $c_A$ be the smallest possible integer such that $c_A \geq x_A$ and $(c_A-1)(2\gamma_A-1) > \theta + 2|E(H)|\cdot(x_A+x_B)$.
\item[(e):] Let $c_B$ be the smallest possible integer such that $c_B \geq x_B$ and $(c_B-1)(1-2\gamma_B) > \theta + 2|E(H)|\cdot(x_A+x_B)$ and
$\frac{c_B-1}{|E(H)|} \geq  \frac{\gamma_B}{1-\gamma_B}$.  
\item[(f):] Let $z= \lceil 3|E(H)| \frac{1-\gamma_B}{\gamma_B(1-2\gamma_B)} \rceil$. Note that $\theta=6|E(H)| + 2z|V(H)|$.

\end{description}

We will now construct a graph $G$ and a partition $(A,B)$ of $V(G)$ such that a solution to our problem for $G$ will give us a minimum transversal in $H$. For an illustration, see Fig. \ref{fig:G}.

Let $C_A$ and $C_B$ be cliques, such that $|C_A|=c_a$ and $|C_B|=c_b$ (see (d) and (e)). 
For each edge $e \in E(H)$, let $r_e$ be a vertex and let $R=\cup_{e \in E(H)} \{r_e\}$.
For each edge $u \in E(H)$, let $u'$ be a vertex and let $V'=\cup_{u \in V(H)} \{u'\}$.
Furthermore let $Z_u$ denote a set of $z$ vertices and let $Z=\cup_{u \in V(H)} Z_u$.
We will now let $V(G)=V(C_A) \cup V(C_B) \cup R \cup V' \cup Z$.

Let $E_1$ denote the edges in $C_A \cup C_B$. For every vertex $r_e \in R$ add $x_A$ edges to $C_A$ and add $x_B$ vertices to $C_B$.
Let $E_2$ denote these edges.
For every $e \in E(H)$ add an edge from $r_e$ to $u'$ if and only if $u \in V(e)$ in $H$. Let $E_3$ denote these edges.
For every $u \in V(H)$ add all edges between $u'$ and $Z_u$ to $G$. Let $E_4$ denote these edges.
Note that $|E_4|=|V(H)|\cdot z = |V(H)|\lceil 3|E(H)|\frac{1-\gamma_B}{\gamma_B(1-2\gamma_B)}\rceil$ (by (f)).
This completes the description of $G$ and note that $E(G)=E_1 \cup E_2 \cup E_3 \cup E_4$.

Note that by (a)-(f) above all values of $\theta$, $x_B$, $x_A$, $c_A$, $c_B$ and $z$ are polynomial in $|V(H)|+|E(H)|$ (as $\gamma_A$ and $\gamma_B$ are considered 
constant). Therefore, $|V(G)|$ is also a polynomial in $|V(H)|+|E(H)|$. 
Let all vertices in $C_A$ belong to $A$ and all other vertices belong to $B$. That is $B = V(C_B) \cup R \cup V' \cup Z$.
This completes the definition of the partition $(A,B)$ of $V(G)$.

\begin{figure}[th]
  \begin{center}
\tikzstyle{vertexB}=[circle,draw, minimum size=10pt, scale=0.7, inner sep=0.9pt]
\tikzstyle{vertexA}=[rectangle,draw, minimum size=8pt, scale=0.9, inner sep=0.9pt]
\tikzstyle{vertexBs}=[circle,draw, minimum size=6pt, scale=0.5, inner sep=0.9pt]
\begin{tikzpicture}[scale=0.5]


  \draw (0,12) rectangle (8,14); \node at (4,13) {$C_A$};
  \node (a1) at (1,13.2) [vertexA]{};
  \node (a2) at (2,12.8) [vertexA]{};
  \node (a3) at (6,13.2) [vertexA]{};
  \node (a4) at (7,12.8) [vertexA]{};

  \draw (10,12) rectangle (18,14); \node at (14,13) {$C_B$};
  \node (a1) at (11,13.2) [vertexB]{};
  \node (a2) at (12,12.8) [vertexB]{};
  \node (a3) at (16,13.2) [vertexB]{};
  \node (a4) at (17,12.8) [vertexB]{};

  \node (r1) at (1,9) [vertexB]{$r_{e_1}$};
  \node (r2) at (5,9) [vertexB]{$r_{e_2}$};
  \node (r3) at (9,9) [vertexB]{$r_{e_3}$};
  \node at (13,9) {$\cdots$};
  \node (r4) at (17,9) [vertexB]{$r_{e_m}$};

  \draw[line width=0.03cm] (r1) to (2,12);
  \draw[line width=0.03cm] (r1) to (10.5,12);

  \draw[line width=0.03cm] (r2) to (3,12);
  \draw[line width=0.03cm] (r2) to (11.5,12);

  \draw[line width=0.03cm] (r3) to (4,12);
  \draw[line width=0.03cm] (r3) to (12.5,12);

  \draw[line width=0.03cm] (r4) to (6,12);
  \draw[line width=0.03cm] (r4) to (14.5,12);

  \node (u1) at (1,4) [vertexB]{$u_1'$};
  \node (u2) at (4,4) [vertexB]{$u_2'$};
  \node (u3) at (7,4) [vertexB]{$u_3'$};
  \node (u4) at (10,4) [vertexB]{$u_4'$};
  \node at (13.5,4) {$\cdots$};
  \node (u5) at (17,4) [vertexB]{$u_n'$};

  \draw[line width=0.03cm] (r1) to (u1);
  \draw[line width=0.03cm] (r1) to (u2);
  \draw[line width=0.03cm] (r1) to (u3);

  \draw[line width=0.03cm] (r2) to (u1);
  \draw[line width=0.03cm] (r2) to (u3);
  \draw[line width=0.03cm] (r2) to (u4);

  \draw[line width=0.03cm] (r3) to (8.6,8);
  \draw[line width=0.03cm] (r3) to (9,8);
  \draw[line width=0.03cm] (r3) to (9.4,8);

  \draw[line width=0.03cm] (r4) to (16.6,8);
  \draw[line width=0.03cm] (r4) to (17,8);
  \draw[line width=0.03cm] (r4) to (17.4,8);

  \draw (0,-1.5) rectangle (2,1); \node at (1,-0.7) {$Z_{u_1'}$};
  \node (z11) at (0.4,0.5) [vertexBs]{};
  \node (z12) at (1,0.5) [vertexBs]{};
  \node (z13) at (1.6,0.5) [vertexBs]{};
  \draw[line width=0.08cm] (u1) to (1,1);

  \draw (3,-1.5) rectangle (5,1); \node at (4,-0.7) {$Z_{u_2'}$};
  \node (z11) at (3.4,0.5) [vertexBs]{};
  \node (z12) at (4,0.5) [vertexBs]{};
  \node (z13) at (4.6,0.5) [vertexBs]{};
  \draw[line width=0.08cm] (u2) to (4,1);

  \draw (6,-1.5) rectangle (8,1); \node at (7,-0.7) {$Z_{u_3'}$};
  \node (z11) at (6.4,0.5) [vertexBs]{};
  \node (z12) at (7,0.5) [vertexBs]{};
  \node (z13) at (7.6,0.5) [vertexBs]{};
  \draw[line width=0.08cm] (u3) to (7,1);

  \draw (9,-1.5) rectangle (11,1); \node at (10,-0.7) {$Z_{u_4'}$};
  \node (z11) at (9.4,0.5) [vertexBs]{};
  \node (z12) at (10,0.5) [vertexBs]{};
  \node (z13) at (10.6,0.5) [vertexBs]{};
  \draw[line width=0.08cm] (u4) to (10,1);

  \draw (16,-1.5) rectangle (18,1); \node at (17,-0.7) {$Z_{u_n'}$};
  \node (z11) at (16.4,0.5) [vertexBs]{};
  \node (z12) at (17,0.5) [vertexBs]{};
  \node (z13) at (17.6,0.5) [vertexBs]{};
  \draw[line width=0.08cm] (u5) to (17,1);



 \end{tikzpicture}
\caption{The graph $G$ when $H$ is a $3$-uniform hypergraph with $m$ edges, including the edges $e_1=\{u_1,u_2,u_3\}$ and $e_2=\{u_1,u_3,u_4\}$.
The square vertices in $C_A$ denotes A-vertices and round vertices (everywhere else) denote B-vertices. Furthermore, the thick edges between
$u_i$ and $Z_{u_i}$ denotes that all edges are present.} \label{fig:G}
\end{center} \end{figure}
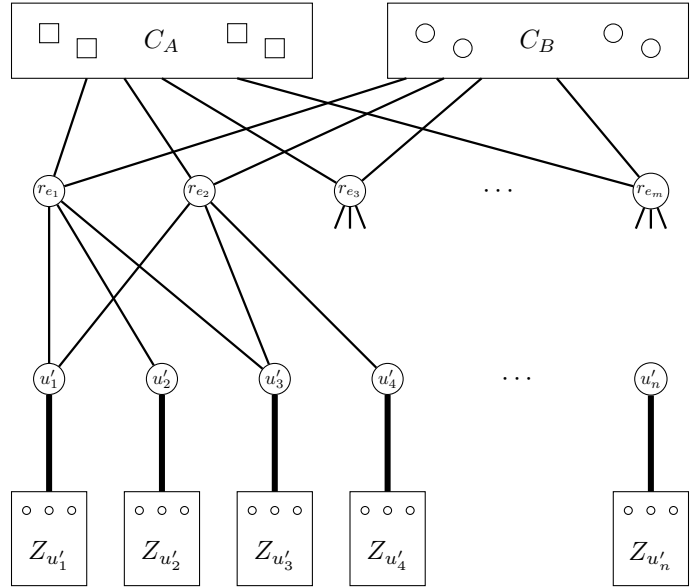

We will show that a welfare optimizing Nash equilibrium for $G$ gives us a minimum transversal for $H$, thereby proving the desired NP-hardness result. 
Assume we have a partition $(X_{\one},X_{\two})$ of $V(G)$ (where
all players in $X_{\one}$ choose action $\one$ and all players in $X_{\two}$ choose action $\two$).

For every edge $uv \in G$ let the weight $w(uv)$ of $uv$ be the increase in welfare achieved because of that edge. That is, 
the following holds.
\[
w(uv) = \left\{ 
\begin{array}{rclcl}
2\gamma_A & & \mbox{if $u,v \in X_{\one}$ and $u,v \in A$} \\
2\gamma_B & & \mbox{if $u,v \in X_{\one}$ and $u,v \in B$} \\
\gamma_A + \gamma_B  & & \mbox{if $u,v \in X_{\one}$ and $|\{u,v\} \cap A|=1$} \\
2(1-\gamma_A) & & \mbox{if $u,v \in X_{\two}$ and $u,v \in A$} \\
2(1-\gamma_B) & & \mbox{if $u,v \in X_{\two}$ and $u,v \in B$} \\
2-\gamma_A - \gamma_B  & & \mbox{if $u,v \in X_{\two}$ and $|\{u,v\} \cap A|=1$} \\
0 & \hspace{0.3cm} & otherwise \\
\end{array}
\right.
\]
Define $W^*$ be the the value of the optimal welfare of $G-E_3-E_4$.
We will now prove the following claims.

{\bf Claim A:} The optimal welfare of $G-E_3-E_4$ (with value $W^*$) is obtained for the partition $(X_{\one},X_{\two})$ if, and only if, the
following holds.

\begin{description}
 \item[(A)] $V(C_A) \cup R \subseteq X_{\one}$ and $V(C_B) \subseteq X_{\two}$.
\end{description}

Furthermore any partition that does not satisfy (A) above will have welfare at most $W^* - \theta$ (see (a)).

\2

{\bf Proof of Claim A:} Let $C^+$ be the component of $G-E_3-E_4$ containing the vertices $V(C_A) \cup V(C_B) \cup R$.
Let $W^+$ be the value of the welfare for $C^+$ for the partition $(V(C_A) \cup R,V(C_B))$ and let $(X_{\one},X_{\two})$
be any partition  $G-E_3-E_4$.
Let $W_A^+$ be the optimal welfare for $C_A$ and let  $W_B^+$ be the optimal welfare for $C_B$. 

Note that $\min\{2\gamma_A(c_A-1), c_A(c_A-1)(2 \gamma_A-1)\} \geq (c_A-1)(2 \gamma_A-1)$ as $2\gamma_A>1 \geq 2 \gamma_A-1$.
By Lemma~\ref{lemmaAclique} we note that if $V(C_A) \not\subseteq X_{\one}$ then the welfare of $C_A$ is at most
the following (by (d)):

\[
\begin{array}{rcl}
W_A^+ - \min\{2\gamma_A(c_A-1), c_A(c_A-1)(2 \gamma_A-1)\} & \leq &  W_A^+ - (c_A-1)(2 \gamma_A-1) \\
& < & W_A^+ - (\theta + 2|E(H)|\cdot(x_A+x_B)) \\
\end{array}
\]

Analogously, by Corollary~\ref{corBclique} and (e) we note that if  $V(C_B) \not\subseteq X_{\two}$ then the welfare of $C_B$ is at most
the following:

\[
\begin{array}{rcl}
W_B^+ - \min\{2(1-\gamma_B)(c_B-1), c_B(c_B-1)(1-2\gamma_B)\} & \leq &  W_B^+ - (c_B-1)(1-2 \gamma_B) \\
& < &  W_B^+ - (\theta + 2|E(H)|\cdot(x_A+x_B)) \\
\end{array}
\]

As the maximum welfare we can obtain from the edges in $E_2$ is not more than $2|E(H)|\cdot(x_A+x_B)$ we note that if $V(C_A) \not\subseteq X_{\one}$
or $V(C_B) \not\subseteq X_{\two}$ then the welfare of $(X_{\one},X_{\two})$ is at least $\theta$ less than $W_A^+ + W_B^+$ and therefore also 
at least $\theta$ less than $W^*$. So we may assume that $V(C_A) \subseteq X_{\one}$ and $V(C_B) \subseteq X_{\two}$.

Let $r_e \in R$ be arbitrary and assume that $r_e \in X_{\two}$. By moving $r_e$ to $X_{\one}$ we would increase the welfare of $G-E_3-E_4$ by 
$x_A (\gamma_A + \gamma_B) - x_B (2(1-\gamma_B))$. By (c) we note that $(x_B+3)\frac{1-\gamma_B}{\gamma_B} - \frac{1}{\gamma_B} < x_A$,
which implies the following (by (b)):

\[
\begin{array}{rcl}
x_A (\gamma_A + \gamma_B) - x_B (2(1-\gamma_B)) & \geq & \left( (x_B+3)\frac{1-\gamma_B}{\gamma_B} - \frac{1}{\gamma_B} \right) (\gamma_A + \gamma_B) - 2x_B (1-\gamma_B)) \\
& = & x_B \left( \frac{1-\gamma_B}{\gamma_B}(\gamma_A + \gamma_B) - 2(1-\gamma_B)) \right)    + \left( 3 \frac{1-\gamma_B}{\gamma_B}  - \frac{1}{\gamma_B} \right) (\gamma_A + \gamma_B) \\
& = & x_B (1- \gamma_B)\left( \frac{\gamma_A+\gamma_B}{\gamma_B} - 2 \right)    + \left( \frac{2-3\gamma_B}{\gamma_B} \right) (\gamma_A + \gamma_B) \\
& \geq & x_B (1- \gamma_B)\left( \frac{\gamma_A-\gamma_B}{\gamma_B} \right)  \\
& > & \theta \\
\end{array}
\]

This completes the proof of Claim~A.

\2

{\bf Claim B:} Let $(X_{\one},X_{\two})$ be any Nash equilibrium of $G$ that satisfies condition (A) in Claim~A, then condition (B) below holds.

\begin{description}
 \item[(B)] For every $e \in E(H)$ the vertex $r_e$ has a neighbour in $V' \cap X_{\one}$.
\end{description}

\2

{\bf Proof of Claim B:} Recall that by (A) in Claim~A we have  $V(C_A) \cup R \subseteq X_{\one}$ and $V(C_B) \subseteq X_{\two}$.
For the sake of contradiction assume that there exists a $r_e \in R$ such that all three neighbours of $r_e$ in $V'$ belong to $X_{\two}$. 
We will now show that $r_e$ does not satisfy the Nash equilibrium property.

The vertex $r_e$ currently has $x_A$ neighbours in $X_{\one}$ and $x_B+3$ neighbours in $X_{\two}$. So its current welfare is
$x_A \gamma_B$, but if we move $r_e$ to $X_{\two}$ then its welfare will become $(x_B+3)(1-\gamma_B)$. By (c), the following holds.

\[
(x_B+3)(1-\gamma_B) - x_A \gamma_B  >  (x_B+3)(1-\gamma_B) - \left((x_B+3)\frac{1-\gamma_B}{\gamma_B}\right) \gamma_B = 0 
\]

So, $r_e$ does not satisfy the Nash equilibrium property, a contradiction. This completes the proof of Claim~B.

\2

{\bf Claim C:} Let $(X_{\one},X_{\two})$ be any partition of $V(G)$ that satisfies condition (A) in Claim~A, condition (B) in Claim~B and condition (C) below.
Then $(X_{\one},X_{\two})$ is a Nash equilibrium.

\begin{description}
 \item[(C)] For every $u \in V(H)$, either $\{u'\} \cup Z_u \subseteq X_{\one}$ or $\{u'\} \cup Z_u \subseteq X_{\two}$.
\end{description}

\2

{\bf Proof of Claim C:} As vertices in $Z$ are only adjacent to vertices in the same partite set in $(X_{\one},X_{\two})$ as themselves they
do satisfy the Nash equilibrium condition. By Lemma~\ref{NashEq} we therefore need to show that the following holds for all $v_A \in V(C_A)$ and all $v_B \in V(C_B)$ and
all $r_e \in R$ and all $u' \in V'$ and all $z \in Z$.

\begin{itemize}
\item $\frac{|N(v_A) \cap X_{\one}|}{|N(v_A) \cap X_{\two}|} \geq \frac{1-\gamma_A}{\gamma_A}$ or $N(v_A) \cap X_{\two}=\emptyset$ (as $v_A \in A \cap X_{\one}$).
\item $\frac{|N(v_B) \cap X_{\two}|}{|N(v_B) \cap X_{\one}|} \geq \frac{\gamma_B}{1-\gamma_B}$ or $N(v_B) \cap X_{\one}=\emptyset$ (as $v_B \in B \cap X_{\two}$).
\item $\frac{|N(r_e) \cap X_{\one}|}{|N(r_e) \cap X_{\two}|} \geq \frac{1-\gamma_B}{\gamma_B}$ or $N(r_e) \cap X_{\two}=\emptyset$ (as $r_e \in B \cap X_{\one}$).
\item If $u' \in X_{\one}$, then $\frac{|N(u') \cap X_{\one}|}{|N(u') \cap X_{\two}|} \geq \frac{1-\gamma_B}{\gamma_B}$ or $N(u') \cap X_{\two}=\emptyset$ (as $u' \in B \cap X_{\one}$).
\item If $u' \in X_{\two}$, then $\frac{|N(u') \cap X_{\two}|}{|N(u') \cap X_{\one}|} \geq \frac{\gamma_B}{1-\gamma_B}$ or $N(u') \cap X_{\one}=\emptyset$ (as $u' \in B \cap X_{\two}$).
\end{itemize}

The first equation holds as $N(v_A) \cap X_{\two}=\emptyset$. The second equation holds, due to the following (see (e)):
\[
\frac{|N(v_B) \cap X_{\two}|}{|N(v_B) \cap X_{\one}|} \geq  \frac{|V(C_B)|-1}{|E(H)|} \geq  \frac{\gamma_B}{1-\gamma_B} 
\]
The third equation holds because of the following (see (c)):
\[
\begin{array}{rcl}
\frac{|N(r_e) \cap X_{\one}|}{|N(r_e) \cap X_{\two}|} 
& \geq & \frac{x_A+1}{x_B+2} \\
& > & \frac{(x_B+3)\frac{1-\gamma_B}{\gamma_B} - \frac{1}{\gamma_B} +1}{x_B+2} \\
& = & \frac{(x_B+3)(1-\gamma_B) - 1 + \gamma_B}{\gamma_B(x_B+2)} \\
& = & \frac{1-\gamma_B}{\gamma_B} \\
\end{array}
\]

If $u' \in V' \cap X_{\one}$ then all neighbours of $u'$ belong to $X_{\one}$, so $N(u') \cap X_{\two}=\emptyset$, which satisfies the fourth equation.
Let $u' \in V' \cap X_{\two}$ be arbitrary.
By (f) we note that the following holds (as $\gamma_B<1/2$).
\[
z = \lceil 3|E(H)| \frac{1-\gamma_B}{\gamma_B(1-2\gamma_B)} \rceil > |E(H)| \frac{\gamma_B}{1-\gamma_B}+1 
\]
 The fifth equation now holds because of the following:
\[
\frac{|N(u') \cap X_{\two}|}{|N(u') \cap X_{\one}|} \geq \frac{z-1}{|E(H)|} > \frac{\gamma_B}{1-\gamma_B}
\]
This completes the proof of Claim~C.

\2

Let $T$ be any transversal of $H$.
Define a partition $(X_{\one},X_{\two})$ of $V(G)$ as follows. 
For every $u \in T$ let $V(Z_u) \cup \{u'\}$ belong to $X_{\one}$ and for 
every $u \in V(H) \setminus T$ let $V(Z_u) \cup \{u'\}$ belong to $X_{\two}$.
Then add $V(C_A) \cup R$ to $X_{\one}$ and add $V(C_B)$ to $X_{\two}$. 
We call this the {\em $G$-extension} of $T$.

\2

{\bf Claim D:} The $G$-extension of $T$ is a Nash equilibrium with the following welfare, where $\epsilon_T$ is some value satisfying $0 \leq \epsilon_T  < 2z (1 - 2\gamma_B)$.
\[
W^* + \epsilon_T +|V(H)|\cdot z (2(1-\gamma_B)) - |T| \cdot  2z (1 - 2\gamma_B)
\]

\2

{\bf Proof of Claim D:} Let $(X_{\one},X_{\two})$ be the $G$-extension of a transversal $T$ in $H$.
As conditions (A), (B) and (C) are all satisfied in Claims~A, B and C, respectively we note that by Claim~C 
$(X_{\one},X_{\two})$ is a Nash equilibrium.

By Claim~A we note that $E_1 \cup E_2$ contribute $W^*$ towards the welfare of $(X_{\one},X_{\two})$.
Let $\epsilon_T$ be the welfare that the edges $E_3$ contribute and note that $0 \leq \epsilon_T \leq 2|E_3| = 6|E(H)|$. 
By (f), $z= \lceil 3|E(H)| \frac{1-\gamma_B}{\gamma_B(1-2\gamma_B)} \rceil > \frac{3|E(H)|}{1-2\gamma_B} $, which implies that
$0 \leq \epsilon_T \leq 6|E(H)| < 2z (1 - 2\gamma_B)$. 

We will now compute the contribution to the welfare from the edges in $E_4$. Let $u' \in V'$ be arbitrary.
If $u' \in X_{\one}$ then the edge between $u'$ and $Z_u$ contribute $z(2\gamma_B)$ to the welfare.
If $u' \in X_{\two}$ then the edge between $u'$ and $Z_u$ contribute $z(2(1-\gamma_B))$ to the welfare.
As there are $|T|$ vertices in $X_{\one} \cap V'$ we therefore get the following welfare of $(X_{\one},X_{\two})$.
\[
W^* + \epsilon_T +(|V(H)|-|T|)\cdot z (2(1-\gamma_B)) + |T| \cdot  2z\gamma_B
= W^* + \epsilon_T +|V(H)|\cdot z (2(1-\gamma_B)) - |T| \cdot  2z(1-2\gamma_B)
\]
This completes the proof of Claim~D.


\2

{\bf Claim E:} If we can find the welfare optimum Nash equilibrium of $G$ then we can find the minimum transversal in $H$. This implies that 
finding the welfare optimum Nash equilibrium for $\gamma_A$ and $\gamma_B$ is \NP-hard.

\2

{\bf Proof of Claim E:} Assume that $(X_{\one}^{opt},X_{\two}^{opt})$ is a welfare optimal Nash equilibrium of $G$ and let
$w(X_{\one}^{opt},X_{\two}^{opt})$ be the welfare of $(X_{\one}^{opt},X_{\two}^{opt})$. Let $T^{opt}$ be a minimum transversal of $H$.
We will show that the following holds, which completes the proof of Claim~E.
\[
|T^{opt}| = \left\lceil \frac{ W^* + |V(H)|\cdot z (2-2\gamma_B)) - w(X_{\one}^{opt},X_{\two}^{opt}) }{2z(1-2\gamma_B)} \right\rceil
\]
Let $(X_{\one},X_{\two})$ be the $G$-extension of $T^{opt}$, and note that, by Claim~D, $(X_{\one},X_{\two})$ is a Nash equilibrium satisfying the following,
where $0 \leq \epsilon_{T^{opt}} < 2z(1-2\gamma_B)$.
\[
W(X_{\one},X_{\two}) = W^* + \epsilon_{T^{opt}} +|V(H)|\cdot z (2-2\gamma_B)) - |T^{opt}| \cdot  2z(1-2\gamma_B)
\]
This implies the following as $\epsilon_{T^{opt}} < 2z(1-2\gamma_B)$ and $w(X_{\one}^{opt},X_{\two}^{opt}) \geq w(X_{\one},X_{\two})$.
\[
\begin{array}{rcl} \vspace{0.2cm}
|T^{opt}| & = & \frac{ W^* + \epsilon_{T^{opt}} + |V(H)|\cdot z (2-2\gamma_B)) - w(X_{\one},X_{\two}) }{2z(1-2\gamma_B)} \\ \vspace{0.2cm}
& = &  \left\lceil \frac{ W^* + |V(H)|\cdot z (2-2\gamma_B)) - w(X_{\one},X_{\two}) }{2z(1-2\gamma_B)} \right\rceil \hspace{2.5cm} (*) \\
& \geq  & \left\lceil \frac{ W^* + |V(H)|\cdot z (2-2\gamma_B)) - w(X_{\one}^{opt},X_{\two}^{opt}) }{2z(1-2\gamma_B)} \right\rceil \\
\end{array}
\]
By Claim~D $w(X_{\one}^{opt},X_{\two}^{opt}) \geq w(X_{\one},X_{\two}) > W^*$. 
If conditions (A) in Claim~A is not satisfied for $(X_{\one}^{opt},X_{\two}^{opt})$ then by Claim~A the welfare of the edges $E_1 \cup E_2$ is at most $W^* - \theta$. 
This implies the following (by (f)).
\[
\begin{array}{rcl}
w(X_{\one}^{opt},X_{\two}^{opt}) & \leq & W^* - \theta + 2(|E_3|+|E_4|) \\
 &  =  & W^* - (6|E(H)| + 2z|V(H)|) +  2(3|E(H)| + z \cdot |V(H)|) \\
 &  = &  W^* \\
\end{array}
\]
This contradiction implies that condition (A) in Claim~A must be satisfied for $(X_{\one}^{opt},X_{\two}^{opt})$.
As $(X_{\one}^{opt},X_{\two}^{opt})$ is a Nash equilibrium, Claim~B implies that condition (B) in Claim~B holds.

Let $Y$ contain all vertices $u \in V(H)$ that satisfy $u' \in V' \cap X_{\one}$. By condition~(B) in Claim~B 
we note that $Y$ is a transversal in $H$. 
Let $(Y_{\one},Y_{\two})$ be the $G$-extension of $Y$, and note that, by Claim~D, $(Y_{\one},Y_{\two})$ is a Nash equilibrium satisfying the following,
where $0 \leq \epsilon_{Y} < 2z(1-2\gamma_B)$:
\[
w(X_{\one}^{opt},X_{\two}^{opt}) \geq W(Y_{\one},Y_{\two}) = W^* + \epsilon_{Y} +|V(H)|\cdot z (2-2\gamma_B)) - |Y| \cdot  2z(1-2\gamma_B)
\]
This implies the following (as $\epsilon_{Y} < 2z(1-2\gamma_B)$),
\[
\begin{array}{rcl} \vspace{0.2cm}
|T^{opt}| \; \leq \; |Y| & = & \frac{ W^* + \epsilon_{Y} + |V(H)|\cdot z (2-2\gamma_B)) - w(X_{\one}^{opt},X_{\two}^{opt}) }{2z(1-2\gamma_B)}  \\
& = & \left\lceil \frac{ W^* + |V(H)|\cdot z (2-2\gamma_B)) - w(X_{\one}^{opt},X_{\two}^{opt}) }{2z(1-2\gamma_B)} \right\rceil \\
\end{array}
\]

By $(*)$ this implies that the following holds, which completes the proof of Claim~E and this theorem.
\[
|T^{opt}| = \left\lceil \frac{ W^* + |V(H)|\cdot z (2-2\gamma_B)) - w(X_{\one}^{opt},X_{\two}^{opt}) }{2z(1-2\gamma_B)} \right\rceil 
\]
\mbox{}\qed{}

\section{Applications of MWDP in Graph Theory and Related Areas}

\newcommand{\EquivPolyMM}[1]{{ {\bf Equivalence to MWDP:} #1}}

\begin{description}
\item[Weighted Partition:] Given a digraph, $D$, find a partition $(X_1,X_2)$ of $V(D)$ such that the following is maximized where $A(X,Y)$ denotes the set of arcs with tail in $X$ and head in $Y$ (where $a,b,c,d$ are given constants).

\[
 a \cdot |A(X_1,X_1)| +  b \cdot |A(X_1,X_2)| +  c \cdot |A(X_2,X_1)| +  d \cdot |A(X_2,X_2)|
\]

\EquivPolyMM{ Let $M = {\small \left[\begin{array}{cc}
a & b \\
c & d \\
\end{array}\right]}$ for every arc. This equivalence goes both ways so we have a dichotomy for this problem.
}

\item[Max Cut:] Given a graph, $G$, find a spanning bipartite subgraph of $G$ with the maximum number of edges.

\EquivPolyMM{ Let $M = {\small \left[\begin{array}{cc}
0 & 1 \\
1 & 0 \\
\end{array}\right]}$ for any orientation of each edge. This equivalence goes both ways so our dichotomy for this problem implies this problem is \NP-hard.
}

\item[Directed Max Cut:] Given a digraph, $D$, find a partition $(X_1,X_2)$ of $V(D)$ with the maximum number of 
arcs from $X_1$ to $X_2$.

\EquivPolyMM{ Let $M = {\small \left[\begin{array}{cc}
0 & 1 \\
0 & 0 \\
\end{array}\right]}$ for every arc. This equivalence goes both ways so our dichotomy for this problem implies this problem is \NP-hard.
}

\item[Closeness to Eulerian:] Given a digraph, $D$, find a partition $(X_1,X_2)$ of $V(D)$ where the difference between the number of arcs from $X_1$ to $X_2$ and  the number of arcs from $X_2$ to $X_1$ is maximized. Note that this value is zero for Eulerian digraphs.

This value can also be shown to be equal to the minimum number of paths that need to be added to $D$ in order to make it Eulerian (so is the number of paths the Chiness Postman Problem Algorithm adds).

\EquivPolyMM{ Let $M = {\small \left[\begin{array}{cc}
0 & 1 \\
-1 & 0 \\
\end{array}\right]}$ for every arc or if we do not allow negative values we can add 1 to all values and get 
 $M' = {\small \left[\begin{array}{cc}
1 & 2 \\
0 & 1 \\
\end{array}\right]}$.  This equvalence goes both ways so our dichotomy for this problem implies this problem is polynomial.

}

\item[Directed Min $(s,t)$-cut:] Given a digraph, $D$, with $s,t \in V(D)$, find a $(s,t)$-partition $(X_1,X_2)$ (i.e., $s \in X_1$ and $t \in X_2$) with the fewest number of arcs from $X_1$ to $X_2$.

This is equivalent to finding the largest number of arc-disjoint paths from $s$ to $t$ (by Menger's Theorem).

\EquivPolyMM{ Let $M = {\small \left[\begin{array}{cc}
1 & 0 \\
1 & 1 \\
\end{array}\right]}$ and 
 $S = {\small \left[\begin{array}{cc}
|E(G)| & 0 \\
0 & 0 \\
\end{array}\right]}$ and 
 $T = {\small \left[\begin{array}{cc}
0 & 0 \\
0 & |E(G)| \\
\end{array}\right]}$. All arcs of $D$ get associated with matrix $M$. We then add a new vertex $s'$ and the arc $s's$ which we associate with matrix $S$. We also add a new vertex $t'$ and the arc $tt'$ which we associate with matrix $T$.
Now the maximum value we can obtain is $3|A(D)|$ minus the size of a minimum $(s,t)$-cut. So by our dichotomy result this is polynomial.
}

\item[Min $(s,t)$-cut:] Given a graph, $G$, with $s,t \in V(G)$, find a $(s,t)$-partition $(X_1,X_2)$ (ie $s \in X_1$ and $t \in X_2$) with the fewest number of edges between $X_1$ and $X_2$.

This is equvalent to finding the fewest number of edge-disjoint paths between $s$ and $t$ (by Menger's Theorem).

\EquivPolyMM{This follows from the directed min $(s,t)$-cut problem (by replacing each edge by a directed $2$-cycle). But it can also be shown to be polynomial directly using our dichotomy result by letting
$M = {\small \left[\begin{array}{cc}
1 & 0 \\
0 & 1 \\
\end{array}\right]}$ and 
 $S = {\small \left[\begin{array}{cc}
|E(G)| & 0 \\
0 & 0 \\
\end{array}\right]}$ and 
 $T = {\small \left[\begin{array}{cc}
0 & 0 \\
0 & |E(G)| \\
\end{array}\right]}$.}

\item[$2$-color-partition:] Given a $2$-edge-colored graph, $G$, find a partition $(X_1,X_2)$ which maximizes the 
sum of the number of edges in $X_1$ of color one and the number of edges in $X_2$ of color two.

\EquivPolyMM{ Let $M_1 = {\small \left[\begin{array}{cc}
1 & 0 \\
0 & 0 \\
\end{array}\right]}$ and let
 $M_2 = {\small \left[\begin{array}{cc}
0 & 0 \\
0 & 1 \\
\end{array}\right]}$.  By associating $M_1$ to any orientation of each edge of color one and 
associating $M_2$ to any orientation of each edge of color two we note that 
our dichotomy implies that this problem is polynomial.
}

\item[Max Average Degree:] Given a graph, $G$, and an integer $k$,  find a vertex set $X \subseteq V(G)$ such that 
the induced subgraph $G[X]$ has average degree strictly greater than $k$. 

\EquivPolyMM{ Let $M_1 = {\small \left[\begin{array}{cc}
k & 0 \\
0 & 0 \\
\end{array}\right]}$ and let  
$M_2 = {\small \left[\begin{array}{cc}
0 & 0 \\
0 & 2 \\
\end{array}\right]}$.

We will now create a digraph $D$ as follows.
Initially let $D$ be any orientation of $G$. Then,
for each vertex $u \in V(G)$ add a new vertex, $v_u$, and the arc $u v_u$ to $D$. 
This defines $D$. Note that $|V(D)|=2|V(G)|$.

Associate $M_1$ to the arc $u v_u$ for all $u \in V(G)$.
Associate $M_2$ to all other arcs of $D$.

Assume that $(X_1,X_2)$ is an arbitrary solution for $D$.
We may assume that $u$ and $v_u$ belong to the same set in $(X_1,X_2)$.
Let $W=V(G) \cap X_2$. Let $e(W,W)$ denote the number of edges in $G[W]$.
Then the welfare of $(X_1,X_2)$ is the following.

\[
w(X_1,X_2) = k |V(G)| - k |W| + 2e(W,W) 
\]

So $w(X_1,X_2)>k |V(G)|$ if and only if $2e(W,W)> k |W|$, which is equivalent to $k < \frac{2e(W,W)}{|W|}$. 
As $2e(W,W)/|W| = ( \sum_{u \in W} d_{G[W]}(u) )/|W|$ we note that the average degree in $G[W]$ is greater than 
$k$ if and only if $w(X_1,X_2)>k |V(G)|$.

So, by our dichotomy, the Max average degree problem is polynomial.
}

\item[Max Density:] Given a graph, $G$, find a vertex set $X \subseteq V(G)$ such that the number of edges divided by the number of vertices in
the induced subgraph $G[X]$ is maximum
possible.

\EquivPolyMM{ 
This is polynomial by the above result on the Max average degree problem as $e(X,X)/|X|$ is maximum if and only if $2e(X,X)/|X|$ is maximum. 
}


\item[$2$-color-difference:] We are given a edge-weighted $2$-edge-colored graph, $G$, where $w_i(X)$ denotes the weight of all edges in  a vertex set $X$ of color $i$.
Find a set $X \subseteq V(G)$, which maximizes $w_2(X)-w_1(X)$.

\EquivPolyMM{ Let $M_1 = {\small \left[\begin{array}{cc}
1 & 0 \\
0 & 0 \\
\end{array}\right]}$ and let
 $M_2 = {\small \left[\begin{array}{cc}
-1 & 0 \\
0 & 0 \\
\end{array}\right]}$. 
 
 Note that if we do not want negative weights in $M_2$ we can add 1 to all entries without changing the problem (the value of the result just increases by 1 for every 
arc associated with $M_2$).

 Given a digraph $D$ where every arc, $uv$, is associated with a cost $c_{uv}$ and a matrix $M_{uv} \in \{M_1,M_2\}$ we can reduce to an instance, $G$,
of $2$-color-difference as follows. Costs in $D$ become weights in $G$.

  Let $G$ be the underlying graph of $D$ and for every arc $uv \in A(D)$ we color $uv$ with color 2 in $G$ if $M_{uv}=M_1$ and we color
$uv$ with color 1, otherwise. Let $(X_1,X_2)$ be a solution for $D$.
Then the weight of $(X_1,X_2)$ in $D$ is the sum of all the costs  of arcs in $X_1$ associated with $M_1$ minus the sum of the costs of arcs in $X_1$ associated with $M_2$.
However this is exactly $w_2(X_1)-w_1(X_1)$ in $G$. Thus, an optimal solution for the 2-color-difference problem corresponds to an optimal solution for {\bf MWDP}({$M_1$,$M_2$}).  Therefore, as our reduction is   polynomial, our dichotomy result shows that the 2-color-difference problem is \NP-hard.
}
\end{description}
\end{document}

{\color{red}
PS: Note that $2$-color-difference generalizes {\em independent set}, due to the following.
Let $G$ be any graph and create a new graph, $G^*$, by coloring all edges of $G$ with color 1.
Then for each vertex, $u$, in $G$, we add a new vertex $v_u$ to $G$ and the edge $u v_u$ of color 2.
Now let all edges of color 2 have weight 1 and all edges of color 1 have large (greater than $|V(G)|$) weight. 
Then the optimal value of $w_2(X)-w_1(X)$, is exactly the size of a maximum independent set.  
}
}

\end{description}

\end{document}